\documentclass[10pt,twocolumn,twoside]{JCNtran}
\usepackage[pdftex]{epsfig}
\usepackage[latin1]{inputenc}
\usepackage[T1]{fontenc}

\usepackage{amsfonts}
\usepackage{amssymb}
\usepackage[cmex10]{amsmath}
\usepackage{cite}

\usepackage{algorithm}
\usepackage{algorithmic}

\usepackage[tight,footnotesize]{subfigure}
\usepackage{multirow}
\usepackage{siunitx}

\usepackage{array}
\newcolumntype{L}[1]{>{\raggedright\let\newline\\\arraybackslash\hspace{0pt}}m{#1}}
\newcolumntype{C}[1]{>{\centering\let\newline\\\arraybackslash\hspace{0pt}}m{#1}}
\newcolumntype{R}[1]{>{\raggedleft\let\newline\\\arraybackslash\hspace{0pt}}m{#1}}

\usepackage{lipsum}

\usepackage{color}
\usepackage{soul}
\definecolor{mgreen}{rgb}{0,0.2,0.2}
\definecolor{mblue}{rgb}{0,0,0.8}
\definecolor{mred}{rgb}{0.8,0,0}
\newcommand{\tbirkan}[1]{#1}

\newcommand{\linkl}{\ell}
\newcommand{\nummobile}{N_{m}}
\newcommand{\numlink}{N_{\linkl}}
\newcommand{\numfreq}{N_{f}}

\newcommand{\numFailLinks}{N_{\text{fail}}}

\newcommand{\setFreq}{\mathcal{F}}
\newcommand{\setMobile}{\mathcal{M}}
\newcommand{\setLink}{\mathcal{L}}
\newcommand{\connGraph}{G_{C}}
\newcommand{\intGraph}{G_{I}}
\newcommand{\constraintGraph}{G_{S}}
\newcommand{\intMatrix}{\mathbb{I}}
\newcommand{\intij}[2]{I_{#1#2}}

\newcommand{\tackintij}[2]{T_{#1#2}}

\newcommand{\prx}{P^{\text{Rx}}}
\newcommand{\ptx}{P^{\text{Tx}}}
\newcommand{\polang}[1]{A(#1)}
\newcommand{\prxij}[2]{\prx_{#1#2}}
\newcommand{\rxsens}{T_{s}}
\newcommand{\freqsep}[2]{\Delta_{#1#2}}
\newcommand{\freqsepQuantized}[2]{S_{#1#2}}
\newcommand{\freqshift}{\mbox{$\delta_{f}$}}

\newcommand{\funcAssignment}[2]{\chi(#1, #2)}

\newcommand{\setUsedFreq}{\mbox{$\mathbb{A}$}}
\newcommand{\rangefa}{R}
\newcommand{\bandw}{\mbox{$B$}}
\newcommand{\rangeConstraint}{T_{\rangefa}}
\newcommand{\solnfa}[1]{\mbox{$\mathbb{A}$$_{#1}$}}

\newcommand{\freqsepQuantizedWithFreqs}[4]{S_{#1#2}^{#3#4}}

\newcommand{\ith}{$i$th}

\newtheorem{theorem}{Theorem}

\begin{document}

\title{Frequency Assignment Problem with Net Filter Discrimination Constraints}
\author{H. Birkan Yilmaz, Bon-Hong Koo, Sung-Ho Park, Hwi-Sung Park, Jae-Hyun Ham, and Chan-Byoung Chae
\thanks{This research was supported by the Advanced Core Technology Project (511335-912374101) of the Agency for Defense Development (ADD), Korea.}
\thanks{H.~B.~Yilmaz, B.-H.~Koo, and C.-B.~Chae are with the School of Integrated Technology, Yonsei University, Korea (e-mail: \{birkan.yilmaz, harpeng7675, cbchae\}@yonsei.ac.kr).}%
\thanks{H.-S.~Park and J. H.~Ham are with the Agency for Defense Development, Korea (e-mail: \{7hwisung7, mjhham\}@add.re.kr).}%
\thanks{S.-H.~Park is with OpenSNS, Korea (e-mail: shpark@opensns.co.kr).}
} 

\markboth{submitted to IEEE/KICS JOURNAL OF COMMUNICATIONS AND NETWORKS}{Yilmaz \lowercase{\textit{et al}}.: Frequency Assignment Problem with Net Filter Discrimination Constraints} 

\maketitle

\begin{abstract}
Managing radio spectrum resources is a crucial issue. The frequency assignment problem (FAP) basically aims to allocate, in an efficient manner, limited number of frequencies to communication links. Geographically close links, however, cause interference, which complicates the assignment, imposing frequency separation constraints. The FAP is closely related to the graph-coloring problem and it is an NP-hard problem. \tbirkan{In this paper, we propose to incorporate the randomization into greedy and fast heuristics. As far as being implemented, the proposed algorithms are very straight forward and are without system parameters that need tuned.} The proposed algorithms significantly improve, quickly and effectively, the solutions \tbirkan{obtained by greedy algorithms} in terms of the number of assigned frequencies and the range. \tbirkan{The enhanced versions of proposed algorithms perform close to the lower bounds while running for a reasonable duration.} Another novelty of our study is its consideration of the net filter discrimination effects in the communication model. Performance analysis is done by synthetic and measured data, where the measurement data includes the effect of the real 3-dimensional (3D) geographical features in the Daejeon region in Korea. In both cases, we observe a significant improvement by employing randomized heuristics. 
\end{abstract}

\begin{keywords}
Frequency assignment problem, net filter discrimination, performance evaluation.
\end{keywords}

\section{\uppercase{Introduction}}
\label{sec_introduction}
\PARstart{T}{he} radio spectrum is a scarce and valuable resource.  Thus, managing such a scarce resource effectively is an important issue. In wireless networks, connection links may interfere with each other when they are close to one another with respect to the distance of the geographical locations and the assigned frequencies. To reduce the interference, geographically close links should be assigned frequencies with sufficient separation, which increases the resource requirements~\cite{du2013handbookOC}. The problem of assigning a limited number of radio frequencies to the communication links of a network in such a way that interference requirements are satisfied is the frequency assignment problem (FAP)~\cite{pardalos2013nonlinearAP,mannino2003enumerativeAF,aardal1996branchAC,koster1998partialCS}. The FAP, in its most basic form, has two architectural properties: a limited number of frequency bands and frequency separation constraints due to interference between links.

Formally, FAP can be considered as the collection of a set of nodes, communication links, and frequency bands that need to be assigned to communication links with constraints that stem from the interference caused by the assignment. Each communication link has to be assigned a frequency band such that it will not significantly harm the other communication links. An assignment satisfying these constraints is called a \emph{feasible assignment}. There may be many feasible solutions and the objective of the assignment strategy can be minimizing the number of assigned frequencies. Another useful objective can be minimizing the \emph{range} of the assigned frequencies, i.e., the difference between the largest and smallest assigned frequencies~\cite{koster1998partialCS,flood2013fixedLF}.

The FAP is closely related to the graph-coloring problem in that it is a generalization of that problem. Hence, it is an NP-hard problem, which was discovered by Metzger in 1970~\cite{metzger1970spectrumMT,hale1980frequencyAT}. Therefore, researchers concentrate on developing heuristics to find feasible solutions~\cite{montemanni2010heuristicMT,fischetti2000frequencyAI,maniezzo2000antsHF}. A high-level comparison among several heuristic applications can be found in~\cite{aardal2007modelsAS,hurley1996comparisonOL,lozano2012interferencePA,sharma2014compositeDE,koo2016heuristicsFF}. Applications in the satellite and cellular networks are analyzed in~\cite{graham2008frequencyAM,hurley2000channelAI,wang2015multiobjectiveEA,audhya2013newAT,da2013newMA,qian2015adaptiveSF}.  In~\cite{mannino2003enumerativeAF}, net filter discrimination is taken into account for more realistic communication scenarios with  limited parameters (i.e., the filter effect is only considered for adjacent channels). In \cite{koo2016heuristicsFF}, realistic interference constraints are considered and net filter discrimination is utilized.

In this paper, we propose incorporating randomization into greedy heuristics so that the solution is significantly enhanced. The proposed algorithms are based on graph coloring and prioritizing the highest degree nodes. Our communication model is novel in terms of including a net filter discrimination effect. Moreover, we give a theoretical lower bound for the range of assigned frequencies in terms of Hamiltonian path cost. To verify the proposed system, we carry out extensive numerical simulations based on real geographical features.

The rest of the paper is organized as follows: In Section~\ref{sec_system_model}, we describe the system model including interference graph and channel model. In Section~\ref{sec_problem_definition}, we give the details of the problem within an analytical framework. Section~\ref{sec_algorithms} summarizes the solution methods and algorithms. We present the performance of the proposed algorithms in Section~\ref{sec_performance_analysis}, which is followed by the Conclusion.

\section{\uppercase{System Model}}
\label{sec_system_model}
In this section, we give the details of the system with $\nummobile$ quasi-static mobile nodes with uplink and downlink communication links and coordinates. 

\subsection{Connection and Interference Graphs}
\label{subsec_connection_and_interference_graph}
This system can be represented by a connection graph ${\connGraph=(\setMobile, \setLink)}$ where $\setMobile$ and $\setLink$ are the set of communication nodes and  links, respectively. Each link is an ordered pair of communication nodes such that the order determines the direction of communication (e.g., $\linkl_1 = (M_1,M_2)$) and each node has a coordinate in the geographical area (e.g., $M_1=(x_1,y_1)$). Some of the nodes have wireless communication links. Each of the communication links has a two-way link and each of the nodes has one or zero outgoing/incoming communication links. The goal of FAP is to assign each link a particular frequency from the set of frequency bands $\setFreq$ in a way that the other links' qualities will stay above a required quality while minimizing the number or the range of the assigned frequency bands. Range is defined as the difference between the highest and the lowest frequency that is assigned. 

\begin{figure}
	\centering
	\includegraphics[width=0.99\columnwidth,keepaspectratio]
	{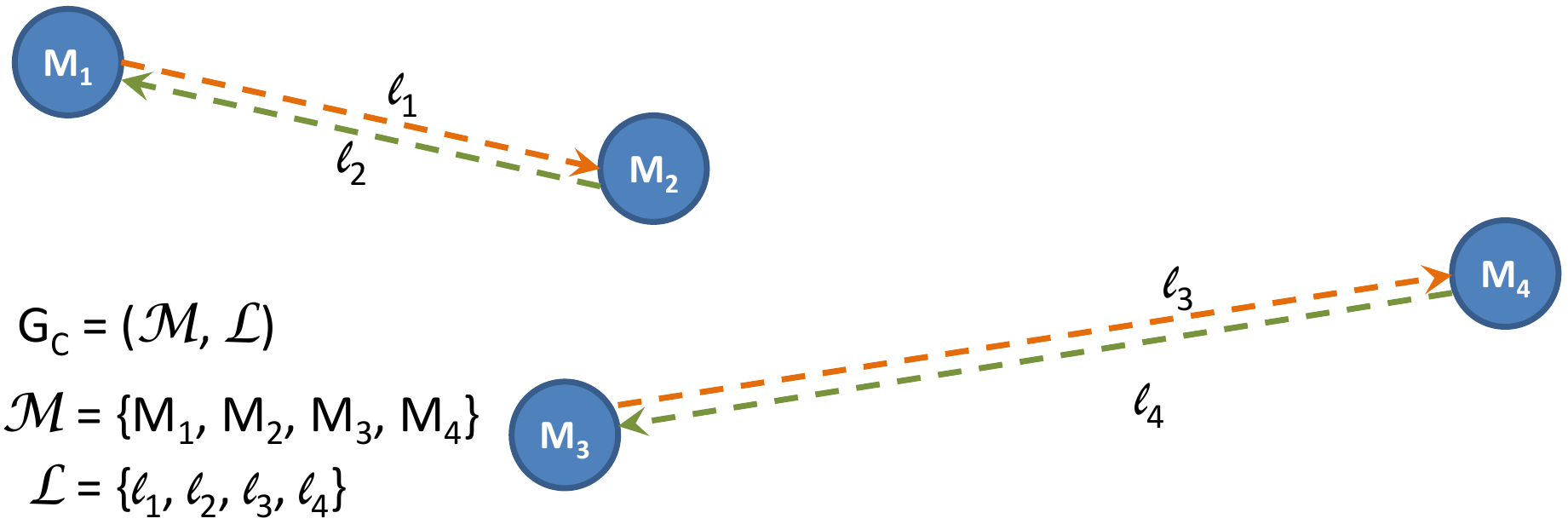}
	\caption{Example topology and system with four nodes.}
	\label{fig_system_topology}
\end{figure}
In Fig.~\ref{fig_system_topology}, we have four nodes (i.e., $\nummobile = 4$) and four communication links (i.e., $\numlink = 4$) between the nodes. When $\linkl_1$ is active (i.e., $M_1$ is transmitting) $\linkl_2$, $\linkl_3$, and $\linkl_4$ are affected with attenuated signals from $M_1$ to the receiving nodes of the links. Evaluating interference from each link to other links results in the interference graph $\intGraph=(V,E)$, where the vertices are the edges of  $\connGraph$ that are corresponding to the set of links (i.e., $V=\setLink$) and all possible link pairs correspond to an interference constraint (i.e., an edge in $E$). Note that the interference constraints are not symmetric. For example, if you consider the interference from $\linkl_1$ to $\linkl_3$, we consider the distance between $M_1$ and $M_4$ while for the interference from $\linkl_3$ to $\linkl_1$ we need to consider the distance between $M_3$ and $M_2$. We define the interference matrix, $\intMatrix$, for $\intGraph$ so that the entries of the matrix correspond to the interference values between links (i.e., $\intij{i}{j}$ is the interference experienced at the receiver of the $\linkl_j$ and caused by the transmitter of the $\linkl_i$).

\subsection{Channel Model}
\label{subsec_channel_model}
ITU-R P.525 model for free space attenuation is used in this paper~\cite{itur1994p525} and given by
\begin{align}
\begin{split}
\prx [dBW]&= \ptx [dBW]+ G - PL  \\
G 	&= 58 - \polang{\phi_t, \phi_r} \\
PL 	&= 32.4 + 20 \log(f \,[MHz]) + 20 \log(d \,[km])
\end{split}
\end{align}
where $\prx$, $\ptx$, $G$, $PL$, $A(.,.)$, $\phi_t$, $\phi_r$, $f$, and $d$ correspond to the received power, transmit power, antenna gain, path loss, antenna pattern function, transmitter antenna angle, receiver antenna angle, frequency, and distance, respectively. The antenna pattern function, which depends on the angle, is illustrated in Fig.~\ref{fig_antenna}, which depends on measurement data.\footnote{We cannot present the measured data for the antenna pattern due to non-disclosure agreement.} We apply our algorithms to both the simplified and the real measured data. The expected interference to tackle (that is from source link $i$ to victim link $j$) is given as
\begin{align}
\begin{split}
\label{eqn_interference_to_tackle}
\tackintij{i}{j} &= \intij{i}{j} - \rxsens = \prxij{i}{j} - \rxsens
\end{split}
\end{align}
where $\prxij{i}{j}$ and $\rxsens$ are the received interference power from link $i$ to $j$ and the receiver sensitivity threshold, which is the required signal strength for a link to be successful.  

\begin{figure}
	\centering
	\includegraphics[width=0.69\columnwidth,keepaspectratio]
	{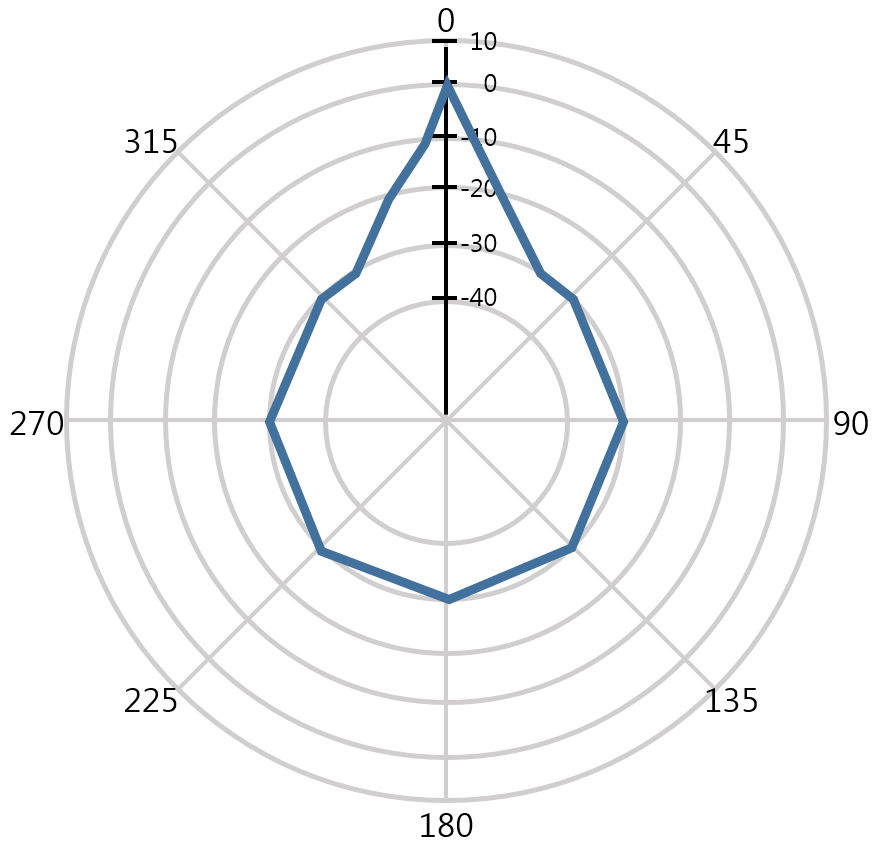}
	\caption{Simplified antenna pattern based on the measured data.}
	\label{fig_antenna}
\end{figure}

At the transmitter and receiver, net filter discrimination (NFD) filtering is utilized for interference. NFD depends upon transmitter spectrum mask and overall receiver filter characteristics. The definition of NFD is given in~\cite{etsi2005tr101} as follows:
\begin{align}
\begin{split}
\text{NFD} (\Delta f)&= 10 \log\left(\frac{P_c}{P_a}\right)\\
P_c &= \int\limits_{0}^{\infty} D(f) \, |H(f)|^2 \,df \\
P_a &= \int\limits_{0}^{\infty} D(f-\Delta f) \, |H(f)|^2 \,df, \\
\end{split}
\end{align}
where $P_c$ denote the total received power after co-channel RF, IF, and base-band filtering, and $P_a$ is the total received power after offset RF, IF, and base-band filtering. $D(f)$, $H(f)$, and $\Delta f$ correspond to the transmitter spectrum mask, overall receiver filter response, and frequency separation between the desired and interference signal, respectively. Note that, $\Delta f = 0$ yields to $0$ dB NFD for the co-channel interference.

\subsection{Frequency Separation}
\label{subsec_freq_sep}
The inverse function of NFD can be utilized for evaluating the pairwise-required frequency separation for the links as follows 
\begin{align}
\begin{split}
\label{eq_freq_sep}
\freqsep{i}{j}  &= \max \{\text{NFD}^{-1}(\tackintij{i}{j}), \, \text{NFD}^{-1}(\tackintij{j}{i})\}
\end{split}
\end{align}
where $\freqsep{i}{j}$ stands for the required frequency separation for the links $i$ and $j$. Note that we have to consider both $\tackintij{i}{j}$ and $\tackintij{j}{i}$ since the interference constraints are not symmetric.

The main goal of the FAP is to assign each link a frequency band from $\setFreq$ while satisfying the frequency separation constraints. Each frequency band has a bandwidth and center frequency that are denoted by $B$ and $f_i$. If the unit frequency band (step size of the overlapping bands) is denoted by $\freqshift$ then the set of frequency bands is defined as
\begin{align}
\begin{split}
\setFreq &= \left\{ \left(f_i\!-\!\frac{B}{2}, \,\, f_i\!+\!\frac{B}{2}\right) \,\, \Big| i=1...\numfreq\right\}\\
f_i & = f_{\text{start}} + \frac{B}{2} +(i-1)\,\, \freqshift  \\
f_{\text{end}} & \geq f_{\numfreq}+\frac{B}{2}
\end{split}
\end{align}
where $\numfreq$, $f_{\text{start}}$ and $f_{\text{end}}$ denote the number of frequency bands, the start and the end frequency for the available range of frequencies. Note that the frequency bands in $\setFreq$ need not to be exclusive but may, \tbirkan{since we use NFD}, intersect. However, the assigned frequency bands should be separated (in terms of center frequencies) according to separation constraints obtained from~\eqref{eq_freq_sep}. \tbirkan{Having overlapping channels increases the number of available frequency bands, thus complicating the solution methods.} In shorthand notation, we use $f_i$ to refer to the frequency band $(f_i-\frac{B}{2}, \,\, f_i+\frac{B}{2})$. 

Due to the discrete nature of $\setFreq$, frequency separation can be done in a quantized manner with multiples of $\freqshift$. Therefore, we convert the frequency separation constraints to quantized frequency separation constraints by 
\begin{align}
\begin{split}
\freqsepQuantized{i}{j} &= \lceil\freqsep{i}{j} / \freqshift \rceil
\end{split}
\end{align}
where $\freqsepQuantized{i}{j}$ corresponds to the required quantized frequency separation between link $\linkl_i$ and $\linkl_j$. These constraints are used in the problem definition and algorithms. To be successful, for example, $\freqsepQuantized{i}{j} = 4$ means the difference of the assigned frequency indices should be at least four for the links $\linkl_i$ and $\linkl_j$.

\section{\uppercase{Problem Definition \& Lower Bound for Range}}
\label{sec_problem_definition}
We define a binary variable that is the assignment variable for the given link and the frequency as 
\begin{align}
\begin{split}
\funcAssignment{\linkl}{f} &=\left\{  
\begin{array}{ll}
1 & \mbox{If $f$ is assigned to the link $\linkl$}\\
0 & \mbox{otherwise}
\end{array} \right.
\end{split}
\end{align}
and a set that contains the assigned frequency indices as
\begin{align}
\begin{split}
\setUsedFreq &= \left\{ k \;\left| \sum\limits_{\linkl} \funcAssignment{\linkl}{f_k} > 0 \right. \right\}
\end{split}
\end{align}
for the problem definitions. Optimization problems aim to minimize the number of used frequencies and the range, $\rangefa$, while considering the frequency separation constraint.

\subsection{Minimizing Number of Used Frequencies}
\label{subsec_min_used_freq}
By using these variables, FAP minimizes the number of used frequencies,  which can be formalized as
\begin{align}
\begin{split}
\mathop{\mbox{Minimize}}\limits_{\funcAssignment{\linkl}{f}} \sum\limits_{f} {\min \Bigg\{1, \,\sum\limits_{\linkl} \funcAssignment{\linkl}{f} \Bigg\}}
\end{split}
\end{align}
subject to
\begin{align}
\label{eq_link_requirement}
\sum\limits_{f} \funcAssignment{\linkl}{f} = 1  & \;\;\; \forall \linkl \\
\label{eq_separation_requirement}
| k\funcAssignment{\linkl_i}{f_k} - m\funcAssignment{\linkl_j}{f_m} | \geq \freqsepQuantizedWithFreqs{i}{j}{f_k}{f_m}  & \;\;\; \forall \linkl_i, \linkl_j, f_k, f_m \\
\rangefa(\setUsedFreq)  & \leq \rangeConstraint \\
\funcAssignment{\linkl}{f} \in \{0,1\}  &\;\;\; \forall \,\linkl, f
\end{align}
where $\setUsedFreq$ is the set of assigned frequencies, $\rangeConstraint$ is the range constraint, $\rangefa(\setUsedFreq) = (\max(\setUsedFreq{}) - \min(\setUsedFreq{})) \,\freqshift{} + \bandw$ is the range of $\setUsedFreq$, and $\freqsepQuantizedWithFreqs{i}{j}{f_k}{f_m} = \freqsepQuantized{i}{j}\,\funcAssignment{\linkl_i}{f_k}\,\funcAssignment{\linkl_j}{f_m}$.

The first constraint in \eqref{eq_link_requirement} guarantees exactly one frequency assignment for each link. The second constraint in \eqref{eq_separation_requirement} enforces the frequency separation requirements. The third constraint concerns the range constraint and guarantees that the range of assigned frequencies is smaller than $\rangeConstraint$.

\subsection{Minimizing Range}
\label{subsec_min_range}
\begin{align}
\begin{split}
\mathop{\mbox{Minimize}}\limits_{\funcAssignment{\linkl}{f}} \;\; \rangefa(\setUsedFreq) 
\end{split}
\end{align}
subject to
\begin{align}
\sum\limits_{f} \funcAssignment{\linkl}{f} = 1  & \;\;\; \forall \linkl \\
| k\funcAssignment{\linkl_i}{f_k} - m\funcAssignment{\linkl_j}{f_m} | \geq \freqsepQuantizedWithFreqs{i}{j}{f_k}{f_m}  & \;\;\; \forall \linkl_i, \linkl_j, f_k, f_m \\
\funcAssignment{\linkl}{f} \in \{0,1\}  &\;\;\; \forall \,\linkl, f
\end{align}

In this optimization problem, the aim is to minimize the range of assigned frequencies while satisfying the separation constraints. 

\subsection{Lower Bound for Number of Used Frequencies}
\tbirkan{
The objective of the FAP is to minimize the number of frequency bands so as to maximize the spectral efficiency. In graph theory, the chromatic number of a graph is defined as the smallest number of colors needed to color the vertices so that adjacent vertices have different colors. The size of the maximum clique of graph $G$ is a lower bound for the chromatic number of $G$ and is denoted by $w(G)$. Hence, we have the following lower bound for the number of assigned frequencies for FAP
\begin{align}
w(G) \leq |\setUsedFreq{}| = \sum\limits_{f} {\min \Big\{1, \,\sum\limits_{\linkl} \funcAssignment{\linkl}{f} \Big\}}
\end{align}
where $|.|$ is the cardinality operator.
}

\tbirkan{While it is NP-hard to enumerate all possible maximum cliques, there is a simpler way does exist to decide whether a graph has a clique of size $k$ or not. The authors in~\cite{modani2008largeMC} suggested an effective method for filtering the nodes that are ineligible for the k-clique. We use their method to find a lower bound for the number of used frequencies.}

\subsection{Lower Bound for Range}
\label{subsec_bound_range}
\tbirkan{In a graph $G$, a Hamiltonian path corresponds to a path which visits each vertex once and only once. Finding the minimum cost Hamiltonian path is referred to as to \emph{Open Traveling Salesman Problem}. Let $H(G)$ denote the weight cost of the shortest Hamiltonian path in $G$. In the literature, $H(G)$ has been shown to be a lower bound for the range of assignments which is denoted by $\rangefa(\setUsedFreq)$ or $\rangefa(\funcAssignment{\linkl}{f}, G)$~\cite{allen1997frequencyAP,smith2000newLB,aardal2003modelsAS}. It can be difficult to calculate $H(G)$ for large constraint graphs, so researchers have considered an alternative bound (based on spanning trees)~\cite{allen1997frequencyAP,aardal2003modelsAS}. The motivation is rooted in each Hamiltonian path being a spanning tree. Let $S(G)$ denote the cost of a minimal weight spanning tree of $G$. So $S(G)$ constitutes a lower bound for $H(G)$ and, fortunately, a greedy algorithm is available for evaluating $S(G)$. If we combine these, we end up with
\begin{align}
 S(G) \leq H(G) \leq \rangefa(\setUsedFreq)
\end{align}
where $\rangefa(\setUsedFreq)$ corresponds to the range of any given frequency assignment. In general, $H(G)$ is a tighter bound, but for large constraint graphs it can be more efficient to use $S(G)$ since it has an effective algorithm that runs for a feasible duration. 
}

For special frequency-separation constraint graphs that are derived from $\freqsepQuantized{i}{j}$, we introduce a new theorem  \tbirkan{to characterize the equality condition of the minimum Hamiltonian path bound.} 

\begin{theorem}
If every triangle in the constraint graph, which is derived from $\freqsepQuantized{i}{j}$, satisfies the triangle inequality, then the  minimum Hamiltonian path cost \tbirkan{is equal to minimum} $\rangefa$.
\end{theorem}
\begin{proof}
Suppose we have a constraint graph $G_S=(\setLink, E)$ derived from $\freqsepQuantized{i}{j}$ and every triangle satisfies the triangle inequality. In other words, if $e_1, e_2$, and  $e_3$ form a triangle on the graph, then the sum of any two edges' costs (i.e., $e_1$ and $e_2$; or $e_1$ and $e_3$; or $e_2$ and $e_3$) is greater than the other edge's cost. 

We denote the vertices of $G_S$ as $V=\{1,2,3,\cdots,n_{v}\}$. We define frequency assignment function $C^{*}\!:\!V\!\rightarrow\!{\setFreq}$ as $C(v_i)=f_i$ if and only if the band with central frequency $f_i$ is assigned to the link having the node $v_i$ as the transmitter. We define $H(G_S)$ as the minimum Hamiltonian path cost and $R(C^{*},G_S)$ as the range of assigned frequencies to the graph $G_S$ when the assignment $C^*$ satisfies the separation constraints. As $C^*$ has a finite number of possible assignments (less than $n_{v}^{N_{f}}$), one of them must be the optimal solution having the minimum range. We denote the optimal one as $C^{opt}$. Now we need to show that $R(C^{opt},G_S)=H(G_S)$.

First, we show that $H(G_S)\!\geq\!{R(C^{opt},G_S)}$ holds. Assume that the Hamiltonian path with the shortest cost starts at vertex $h_1$ and is connected through $h_i$ for the $i^{th}$ vertex and ends at the vertex $h_{n_v}$, where $h_{i}$ is the $i^{th}$ element of a re-ordered sequence of $V$. Let $C^{*}(h_1)=f_1$ and $\forall{1}\!<\!i\!\leq\!{n_v}, C^{*}(h_i)=f_{1}+\Sigma_{k=1}^{i-1}{S_{h_{k}h_{k+1}}}$. To show that $C^{*}$ meets the separation constraints, we choose two arbitrary  vertices, $x$ and $y$ from $V$. It is enough to show that $|C^{*}(x)-C^{*}(y)|\geq{S_{xy}}$. Let $i,~j$ be the numbers so that $h_{i}\!=\!x$ and $h_{j}\!=\!y$. Without loss of generality, we say that $i\!<\!j$.
\begin{align}
\begin{split}
|C^{*}(x)-C^{*}(y)|&=|C^{*}(h_j)-C^{*}(h_i)| \\
&=|\Sigma_{k=i+1}^{j}(C^{*}(h_k)-C^{*}(h_{k-1}))| \\
&=|\Sigma_{k=i+1}^{j}S_{h_{k-1}h_{k}}|\geq{S_{h_{j}h_{i}}}=S_{xy}.
\nonumber
\end{split}
\end{align}
Note that the last inequality holds from the triangular inequality condition. Hence, it is proved that $C^{*}$ is a valid assignment, which implies that $R(C^*,G_S)\geq{R(C^{opt},G_S)}$. 
\begin{align}
\begin{split}
\nonumber
R(C^*,G_S)&\!=\!C^{*}(h_{n_v})\!-\!C^{*}(h_{1})\!=\!\Sigma_{k=1}^{n_v-1}S_{h_{k}h_{k+1}}\!=\!H(G_S).\\ 
&\therefore{H(G_S)\geq{R}(C^{opt},G_S)}.
\end{split}
\end{align}

Now we show that $H(G_S)\!\leq\!{R(C^{opt},G_S)}$ holds. Assume that $C^{*}$ is one valid assignment to $G_S$. $\forall1\!\leq\!i\!\leq\!n_{v}$, $C^{*}(i)$s are given satisfying $|C^*(j)\!-\!C^*(k)|\!\geq\!{S_{jk}}$ for arbitrary $j,~k$ from $V$. Let $a_i$ be a sequence of $V$ such where $C^{*}(a_i)$ values are sorted in ascending order. In other words, $C^{*}(a_i)\!\leq\!C^{*}(a_j)$ holds if $i\!<\!j$ holds. Then the range of the assignment becomes $C^{*}(a_{n_v})\!-\!C^{*}(a_1)$.
\begin{align}
\begin{split}
\label{eqn_RgeqH}
R(C^{*},G_S)&=C^{*}(a_{n_v})\!-\!C^{*}(a_1) \\
&=\Sigma_{k=2}^{n_v}C^{*}(a_k)\!-\!C^{*}(a_{k\!-\!1})\geq{\Sigma_{k=2}^{n_v}S_{a_{k}a_{k\!-\!1}}} \\
&=S_{a_{1}a_{2}}\!+\!S_{a_{2}a_{3}}\!+\!\cdots\!+\!S_{a_{n_{v}\!-\!1}a_{n_v}}\!\geq\!{H}(G_S).
\end{split}
\end{align}
The first inequality of \eqref{eqn_RgeqH} comes from the condition that $C^*$ is a valid assignment, and the second inequality comes from the fact that such a summation of $n_{v}-1$ path costs becomes one of the possible Hamiltonian path costs of $G_S$, which should be larger than the minimum one of them, denoted as $H(G_S)$. The inequality \eqref{eqn_RgeqH} indicates that for any valid assignment $C^*$, $H(G_S)\!\leq\!{R(C^{*},G_S)}$ holds. Since $C^{opt}$ is also one of the valid assignments, we can say that $H(G_S)\!\leq\!{R(C^{opt},G_S)}$ holds.

Finally, $H(G_S)\!\leq\!{R(C^{opt},G_S)}\!\leq\!H(G_S)$ is achieved, which leads to the desired equality $H(G_S)\!=\!{R(C^{opt},G_S)}$.

\end{proof}

Hence, we proved that the minimum Hamiltonian path cost \tbirkan{is equal to the minimum} range for the frequency assignment when the constraint graph satisfies with the triangle inequality. \tbirkan{In other words, the minimum Hamiltonian path cost is exactly the minimum range for the given FAP as long as the constraint graph satisfies the triangle inequality.} 


\section{\uppercase{Algorithms}}
\label{sec_algorithms}
Even if all domains have size two, FAP is reduced to a maximum satisfiability problem, which means it is NP-hard~\cite{koster1998partialCS}. A graph-based dynamic programing algorithm is introduced for graphs with bounded tree width~\cite{koster1999optimalSF} that runs in polynomial time, but it is exponential in the width of tree decomposition. Moreover, even with employing several reduction techniques, due to time and memory requirements, it is hard to solve large instances with dynamic programming~\cite{koster1999optimalSF}.

\tbirkan{Our requirement was to approach the problem with perspectives that consider implementation and time complexity. Since the target applications include military and government infrastructures, the approach should also consider ease of integration while still aiming for an effective improvement in terms of resource utilization. Hence we focused on sequential, greedy, and fast algorithms.} First, we consider a greedy algorithm in which the assignment is carried out in a greedy fashion, giving priority to the highest degree node in the undirected constraint graph that is derived from $\freqsepQuantized{i}{j}$ (i.e., $\constraintGraph$). Then we introduce a coloring-based solution and a hybrid algorithm. Moreover, we enhance the algorithms with randomization. Finally, for comparison, we present a genetic algorithm to solve the FAP. \tbirkan{Enhanced algorithms are sequential in nature and faster than common and complicated heuristics with many parameters such as in the cases of simulated annealing and Tabu searches. Note also that our system contains nearly 4000 overlapping frequency bands in our system. Hence the problem is somewhat unique due to each iteration having a large number of neighbor candidates in the GA and Tabu searches. This makes it a necessity to use sequential and simple algorithms. Moreover, we will also present the performances of the proposed algorithms and the Tabu search in a time limited scenario for the CELAR data set. }

\subsection{Highest Degree-based Greedy Algorithm (HEDGE)}
\label{subsec_hedge}
The main idea of the HEDGE algorithm is to give priority to the highest degree nodes in the undirected constraint graph (i.e., $\constraintGraph$) and while assigning the frequency the best result is chosen without knowing the future assignments. 
\begin{algorithm}
\caption{HEDGE Algorithm Pseudocode}
\begin{algorithmic}[1]
\REQUIRE $\freqsepQuantized{i}{j} \;\;\;\; \forall \, i,j$
\STATE Generate $\constraintGraph$ 
\WHILE{There is unassigned link}
	\STATE Choose a node ($\linkl_i$) with the highest degree in $\constraintGraph$
	\IF{AssignFreqTo($\linkl_i$)}
		\STATE Remove the node and update the graph $\constraintGraph$
	\ELSE
		\STATE Rollback neighbor assignments
	\ENDIF
\ENDWHILE
\end{algorithmic}
\end{algorithm}
Here, the AssignFreqTo($\linkl_i$) is a greedy assignment procedure and given in Algorithm~\ref{alg_freqAssignment}. The AssignFreqTo($\linkl_i$) procedure aims to utilize the used frequencies again and again to reduce the number of used frequencies. It also aims to utilize the smallest indexed frequency to keep the range as small as possible. Due to the greedy nature of the algorithm, we cannot be sure whether the solution is optimal or not. However, it runs in a very feasible time. 

\begin{algorithm}
\caption{AssignFreqTo($\linkl_i$) Pseudocode}
\begin{algorithmic}[1]
\REQUIRE $\linkl_i$
\STATE Sort the frequencies according to usage count
\FOR{Each frequency $f$}
\STATE $f\leftarrow$ next minimum of most frequently used frequency 
\IF{Assigning $f$ does not violate $\freqsepQuantized{i}{j}$ for all $j$}
	\STATE Assign $f$ to $\linkl_i$
\ENDIF
\ENDFOR
\IF{Cannot assign a frequency to $\linkl_i$}
	\STATE return false
\ELSE
	\STATE return true
\ENDIF
\end{algorithmic}
\label{alg_freqAssignment}
\end{algorithm}

\subsection{Coloring-based Greedy Algorithm (COG)}
\label{subsec_cog}
Another heuristic we propose is the coloring-based algorithm that aims to assign the least number of frequencies. The COG algorithm has two phases. One is assigning colors to the links without considering the separation constraint and the other is assigning frequencies to the colors and consequently the links while satisfying the separation constraints. COG pseudocode is presented in Algorithm~\ref{alg_cog}.
\begin{algorithm}
\caption{COG Algorithm Pseudocode}
\begin{algorithmic}[1]
\REQUIRE $\freqsepQuantized{i}{j} \;\;\;\; \forall \, i,j$
\STATE Generate $\constraintGraph$ 
\WHILE{There is unassigned link}
	\STATE Choose a node ($\linkl_i$) with the highest degree in $\constraintGraph$
	\STATE Assign minimum of most frequently used color that is different than the neighbors
	\STATE Remove the node and update the graph $\constraintGraph$
\ENDWHILE
\STATE Order the colors
\FOR{color=1:ncolor}
	\STATE Assign frequency band to color according to $\freqsepQuantized{i}{j}$ 
\ENDFOR
\end{algorithmic}
\label{alg_cog}
\end{algorithm}

Since COG does not consider the $\freqsepQuantized{i}{j}$ constraints in the first phase, it is more relaxed than HEDGE. Therefore, the COG algorithm has a solution of which the number of frequencies are less than or equal to those of HEDGE, with a cost of higher $\rangefa$.

\subsection{Hybrid Algorithm}
\label{subsec_hybrid}
HEDGE and COG both have their own pros and cons. Hence, we developed a hybrid algorithm that utilizes both methods. The main idea is to run the COG algorithm up to some predetermined number of assigned links and continue with the HEDGE algorithm on the remaining nodes.
\begin{algorithm}
\caption{Hybrid Algorithm Pseudocode}
\begin{algorithmic}[1]
\REQUIRE $N_{\text{cog}}$
\STATE Run COG algorithm until $N_{\text{cog}}$ assignments are done
\STATE Update $\constraintGraph$
\STATE Run HEDGE algorithm for the remaining nodes
\end{algorithmic}
\label{alg_hybrid}
\end{algorithm}

\subsection{Ordering Enhancements}
\label{subsec_ordering_enhancements}
The ordering of the nodes for assignment determines the solution quality. Since we cannot consider the future assignments in an effective time complexity, we incorporated randomization to the ordering. In order to add randomization to the HEDGE algorithm, we randomly select one of the second-highest degree nodes for the even-indexed assignments and highest degree nodes for the odd-indexed assignments. Adding the randomization resulted in different solutions and the run was replicated many times. Selecting the best solution improved the original HEDGE algorithm. For the COG algorithm, we select the colors in a randomized order for the assignment phase. 

When the ordering enhancement is applied and the algorithm is run many times, we end up with different solution results where the set of assigned frequencies corresponding to the \ith{} solution is denoted by $\solnfa{i}$. Therefore, we need to sort the solutions and choose the best one in terms of sorting criteria. We use a scoring function, $\psi$, for the assignment solutions as follows
\begin{align}
\begin{split}
\label{eq_scoring_function}
\psi(\solnfa{i}) 
	&= \mbox{BF}\, \frac{|\solnfa{i}|-\min\limits_{k}|\solnfa{k}|}{\max\limits_{k}|\solnfa{k}|-\min\limits_{k}|\solnfa{k}|} \\
    &+ (1-\mbox{BF})\, \frac{\rangefa(\solnfa{i})-\min\limits_{k}\rangefa(\solnfa{k})}{\max\limits_{k}\rangefa(\solnfa{k})-\min\limits_{k}\rangefa(\solnfa{k})}
\end{split}
\end{align}
where $|.|$, $\rangefa(.)$ and BF are the cardinality operator, range operator, and the balancing factor for the range and the number of used frequencies. The scoring function in~\eqref{eq_scoring_function} is the sum of normalized resources; hence we seek the smallest score(s).

\subsection{Genetic Algorithm (GA)}
\label{subsec_genetic}
The main idea of the GA is to iteratively keep the fit individuals (i.e., candidate solutions) and create a new generation that results in an  improvement of the solution. In each generation, the fitness of every individual in the population is evaluated. The more fit individuals are stochastically selected from the current population, and each individual's genome is modified (recombined and randomly mutated) to form a new generation. The new generation of candidate solutions is then used in the next iteration of the algorithm. 

Gene representation of a solution is the ordered list of assigned frequencies for the links (e.g., $[1,\, 45,\, 23,\, 16]$ could be a sample assignment for the topology example in Fig.~\ref{fig_system_topology}). Any individual candidate solution can be, in therms of fitness function, good or bad. We define the fitness function in terms of the number of failed links and used frequencies as follows:
\begin{align}
\text{fitness} &= | \setUsedFreq | (\numFailLinks + \text{modifier}) 
\end{align}
where $\numFailLinks$ denotes the number of links below a certain quality threshold. We use a modifier to cope with the $\numFailLinks=0$ case and to balance the weights of the metrics. More fit individuals have smaller values for the fitness function evaluation. For equality cases, we consider the range of the assignment. \tbirkan{Note that implementing GA may be done in different ways (including order preserving genetic operations). Here for a comparison we consider the most straight forward one.}

\begin{figure}[!t]
	\centering
	\includegraphics[width=0.99\columnwidth,keepaspectratio]
	{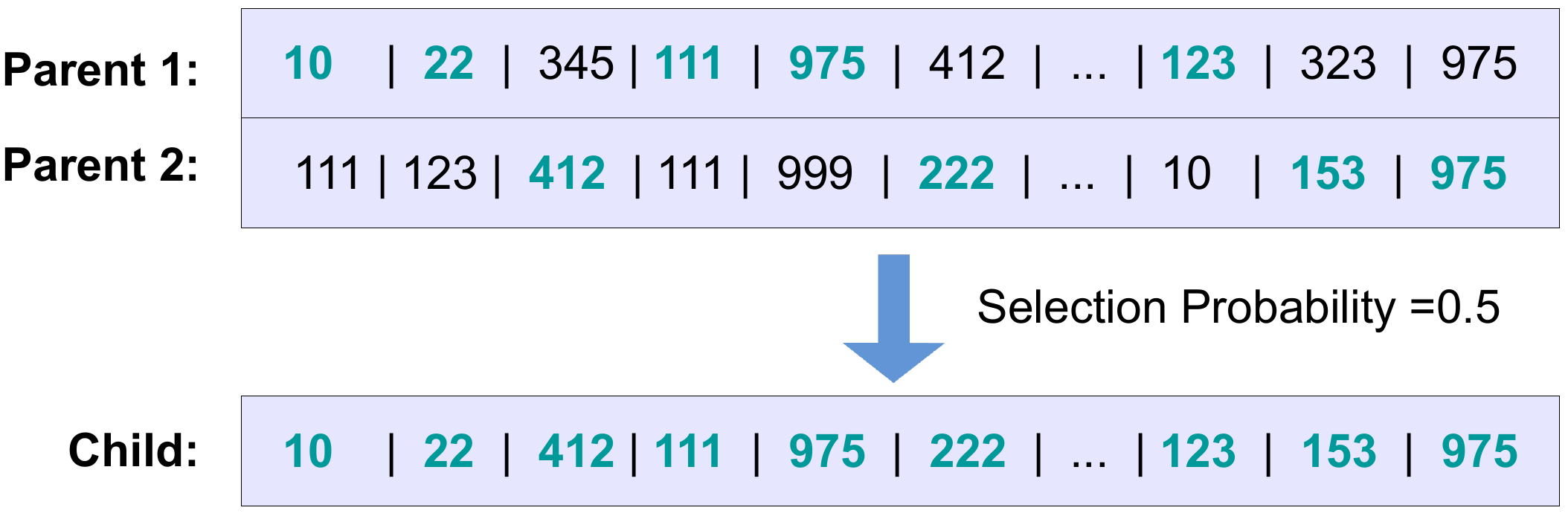}
	\caption{Representative crossover operation.}
	\label{fig_ga_crossover}
\end{figure}
\begin{figure}[t]
	\centering
	\includegraphics[width=0.99\columnwidth,keepaspectratio]
	{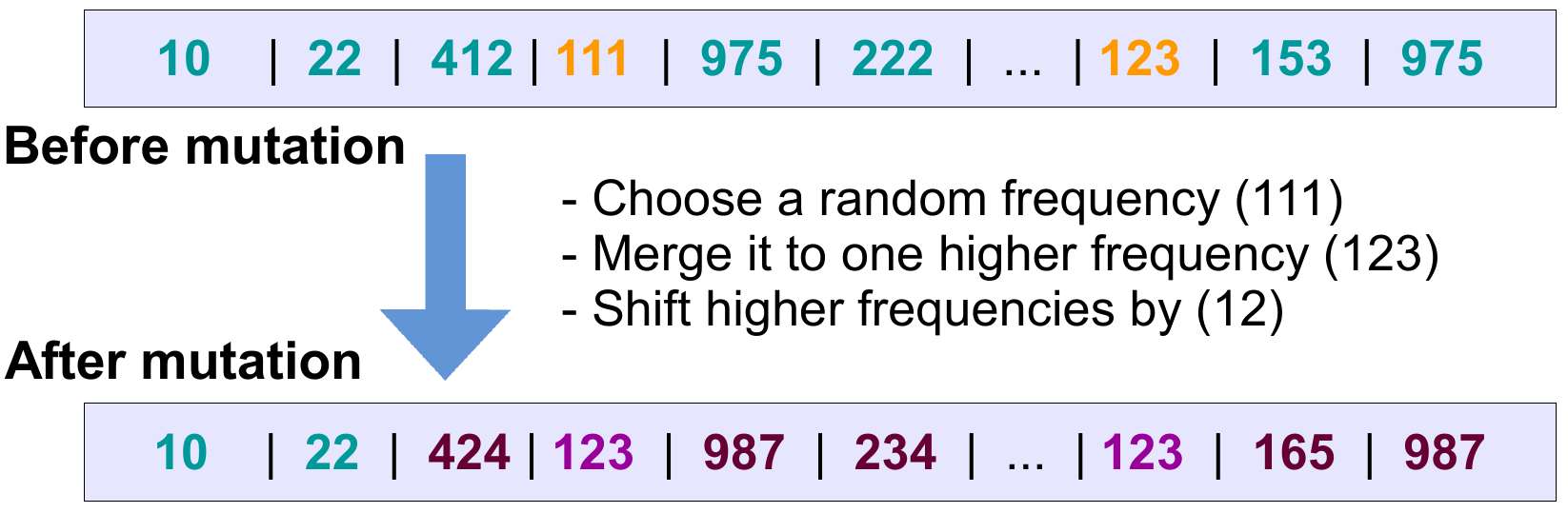}
	\caption{Representative mutation operation.}
	\label{fig_ga_mutation}
\end{figure}
\subsubsection{Initial Population}
Initial population is formed by using the output of the HEDGE algorithm and applying random swap, change, add frequency, and permutation operations. The swap operation switches the assigned frequencies of two random links. The change operation randomly selects a link and changes the assigned frequency to the most and least frequently used frequencies. The add frequency operation randomly selects a link and changes the assigned frequency to a random new frequency that does not violate the constraints. The permutation operation shuffles the assigned frequencies. After creating new individuals by using these operations, GA is iterated.

\subsubsection{Genetic Operations}
The iteration of GA mainly consists of crossover and mutation operations followed by fitness evaluation. During execution of crossover, the best individuals are randomly matched and the genes are combined with equal probability. A representative crossover operations is presented in Fig.~\ref{fig_ga_crossover}, in which a new individual is created by combining the genes of the two parents.

After the crossover operation, randomly selected children of the new generation are exposed to mutation operation, which merge a randomly selected frequency to the next higher frequency used. A representative mutation operation is depicted in Fig.~\ref{fig_ga_mutation}. The mutation operation reduces the number of assigned frequencies at the cost of increased range.

\begin{figure*}[!t]
	\begin{center}
		\subfigure[Top view]
		{\includegraphics[width=0.523\textwidth,keepaspectratio]{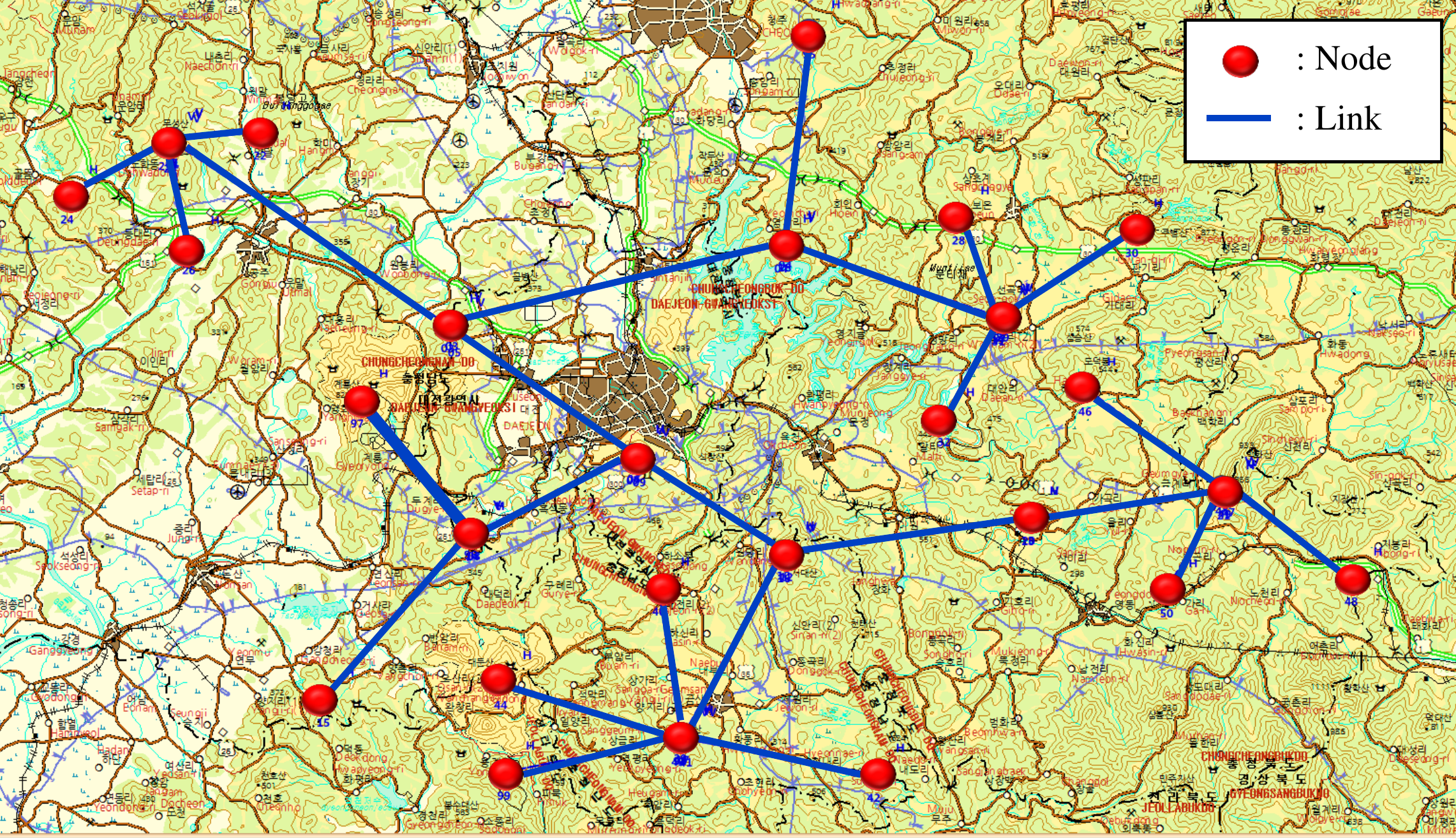}
			\label{fig_top_view} } 
		\subfigure[Front view]
		{\includegraphics[width=0.45\textwidth,keepaspectratio]{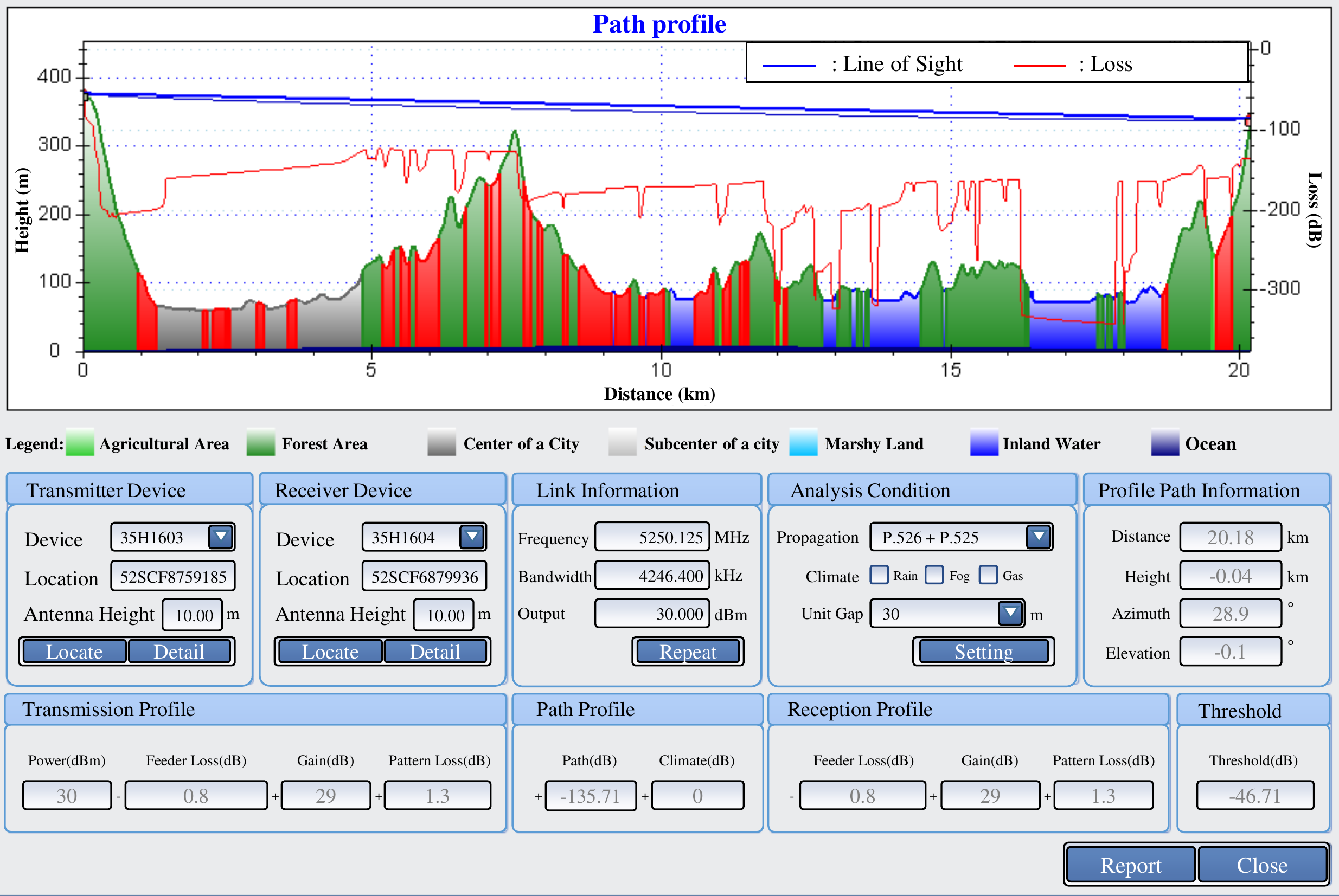}
			\label{fig_front_view} }
	\end{center}
	\caption{An example of topology for analysis with the measured data, covering Daejeon region in South Korea with 50 communication nodes. The topology is considering 3-D environments based on the real geographical features.}
	\label{fig_simulator}
\end{figure*}
\section{\uppercase{Implementation \& Measurement Data}}
\label{sec_implementation}

We measured the interference data between the nodes on actual map, with the real world 3-D ray tracing software. Figure~\ref{fig_simulator} shows 50 communication links that are distributed on the real map of Daejeon, South Korea. The location of nodes are arbitrarily decided due to classified issues, but any two nodes that are communicating with each other are considered to be in a line of sight.  Figure~\ref{fig_top_view} represents the exact distribution of the nodes, while Fig.~\ref{fig_front_view} shows the path environment between the communication nodes, which includes the 3D features of the region.

The software was developed using the programming language C\# based on Microsoft Windows Network. The program has its required input data as the type of antenna polarizations, transmit power, antenna pattern and attenuation values for transmit and receive mask. It also treats the climate and 3D geographical data so as to achieve a more accurate propagation model. It calculates interference between the nodes and power of signal transmission. The calculated values are applied to obtain the requirements for frequency separation. Frequencies are assigned to each node according to the proposed algorithms via MATLAB, satisfying constraints from the measured data to set interference zero between any of the two communication nodes.

We have run our proposed algorithms on several topologies, including 20, 50, and 200-node cases. In this paper we present the figure of 50-node topology and the results for a 50-node and 200-node frequency assignments with the HEDGE, Enhanced HEDGE, Hybrid and Enhanced Hybrid algorithms.

\begin{figure*}[!t]
\begin{center}
\subfigure[$\numlink = 150$]
{\includegraphics[width=0.49\textwidth,keepaspectratio]{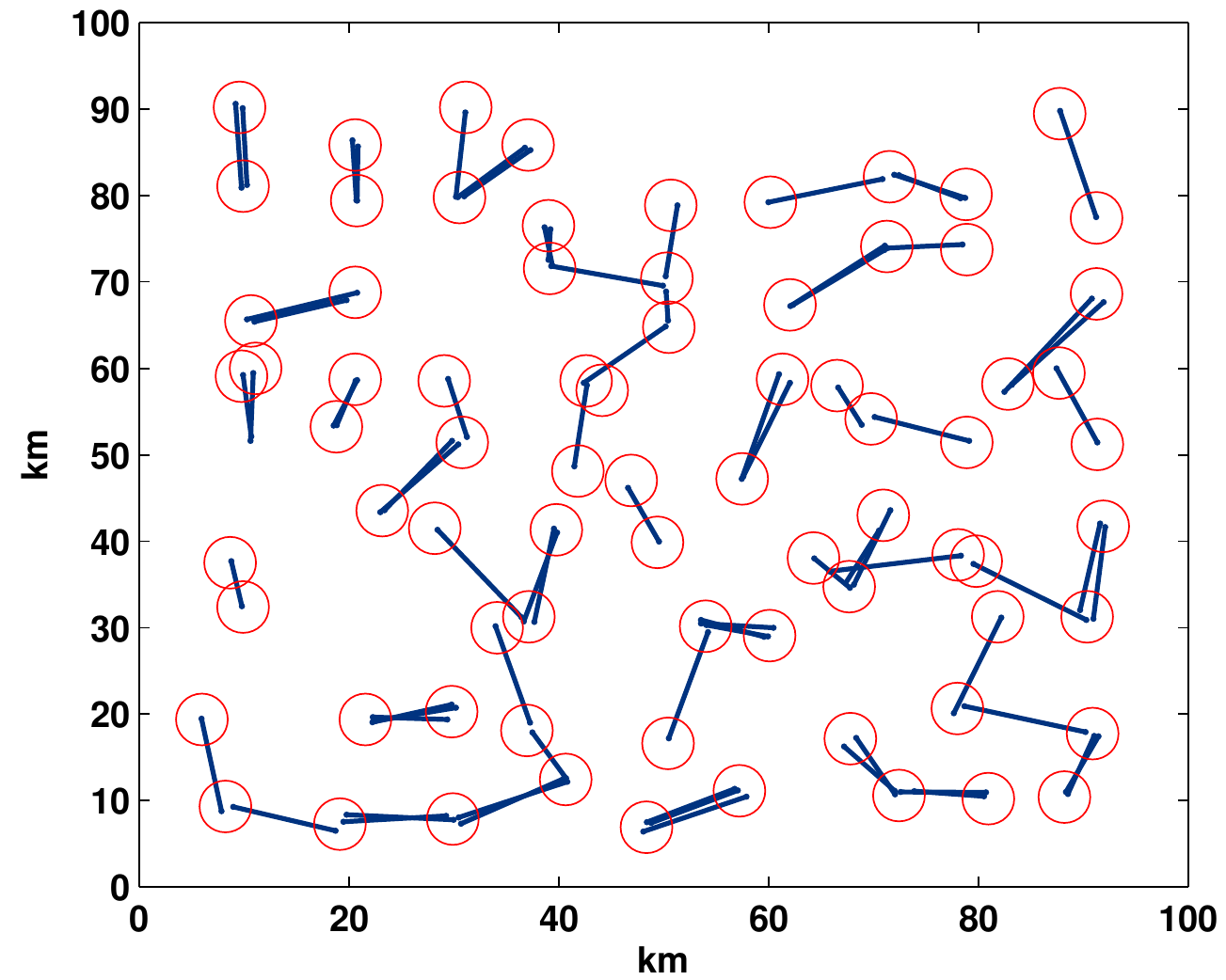}
\label{fig_stopology150} } 
\subfigure[$\numlink = 200$]
{\includegraphics[width=0.49\textwidth,keepaspectratio]{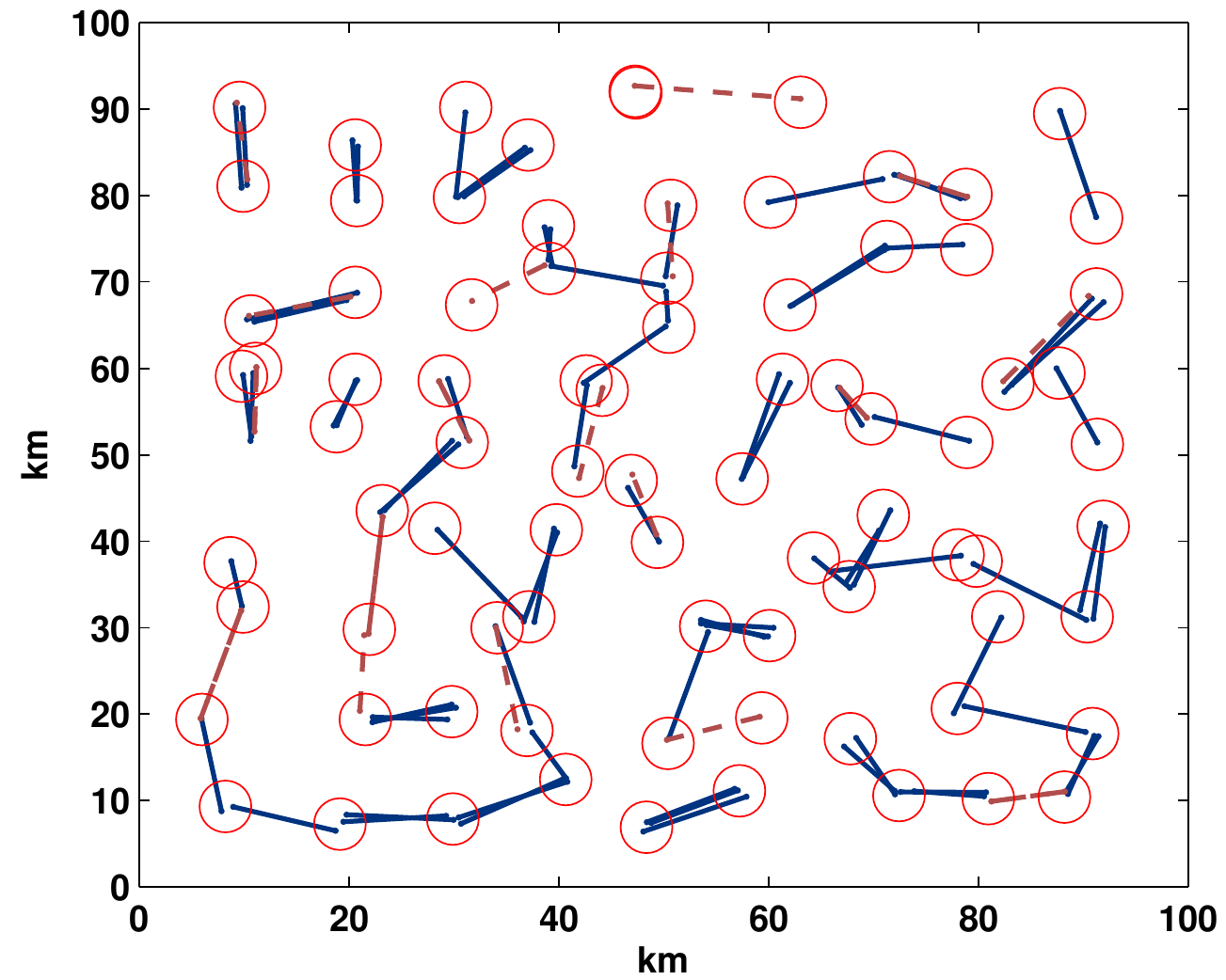}
\label{fig_stopology200} }
\end{center}
\caption{Topology for analysis with synthetic data in a 100 km $\times$ 100 km region. Each shown link has two-way communication links (i.e., each link in the figure corresponds to two links in the system model). For $\numlink=200$ case, there are additionally 50 links, represented with different colors and line style (dashed). Red circles indicate the clusters of communication node locations.}
\label{fig_stopology}
\end{figure*}
\section{\uppercase{Performance Analysis}}
\label{sec_performance_analysis}
In this section, we conduct analyses from two different perspectives: one is via synthetic data and the other is via measurement data. For the analysis via measurement data, we integrated our solution with a software program that utilizes 3-D ray tracing. \tbirkan{Hence the easy integration was a requirement for the developed algorithms}. In both scenarios, we give topology details and we analyze the performance of the proposed algorithms.

\subsection{Analysis with Synthetic Data}
In this section, we analyze the performance of the proposed algorithms and the effect of the balancing factor on the enhanced HEDGE and Hybrid algorithm. We \tbirkan{mainly consider $\numlink=150$} then focus on the scenarios with different numbers of links.

\subsubsection{Topology \& Parameters}
\label{subsec_topology_and_parameters}
Performance of the algorithms are evaluated within an area of 100 km $\times$ 100 km and the maximum link length is 20 km. We consider topologies with 100-200 links within the same physical area, shown in Fig.~\ref{fig_stopology} for $\numlink=150$ and $\numlink=200$. The topology consists of clusters 1 km in size and each cluster includes a number of communication nodes. Other system parameters are listed in Table~\ref{tab_system_params}.

\begin{table}[h]
	\caption{Range of parameters used in the analysis}
	\label{tab_system_params}
	\centering
	\begin{tabular}{L{3.2cm} L{2.2cm} L{2.2cm}}
	\hline
		Parameter &  Variable & Values \\
	\hline
		Link count  		& $\numlink$  	& $100\sim 200$ \\
		Receiver sensitivity& $\rxsens$  	& $-79.12$ dBm \\
		Transmit power		& $\ptx$  		& $1$ W \\
		Range constraint	& $\rangeConstraint$  	& $600$ MHz \\
		Starting frequency	& $f_{\text{start}}$  	& $7007.5$ MHz \\
		Unit frequency band & $\freqshift$  & $150$ KHz \\
		Bandwidth 			& $B$  			& $15$ MHz \\
	\hline
	\end{tabular}
\end{table}

\subfigcapmargin = .15cm
\begin{figure*}[!t]
\begin{center}
\subfigure[Enhanced HEDGE, raw HEDGE, and Genetic algorithms' solution space in terms of range and the number of assigned frequencies. The square and the  circle markers correspond to the solution of the raw HEDGE and the genetic algorithms, respectively. The small points correspond to the solutions obtained by randomization of the HEDGE algorithm.]
{\includegraphics[width=0.49\textwidth,keepaspectratio]{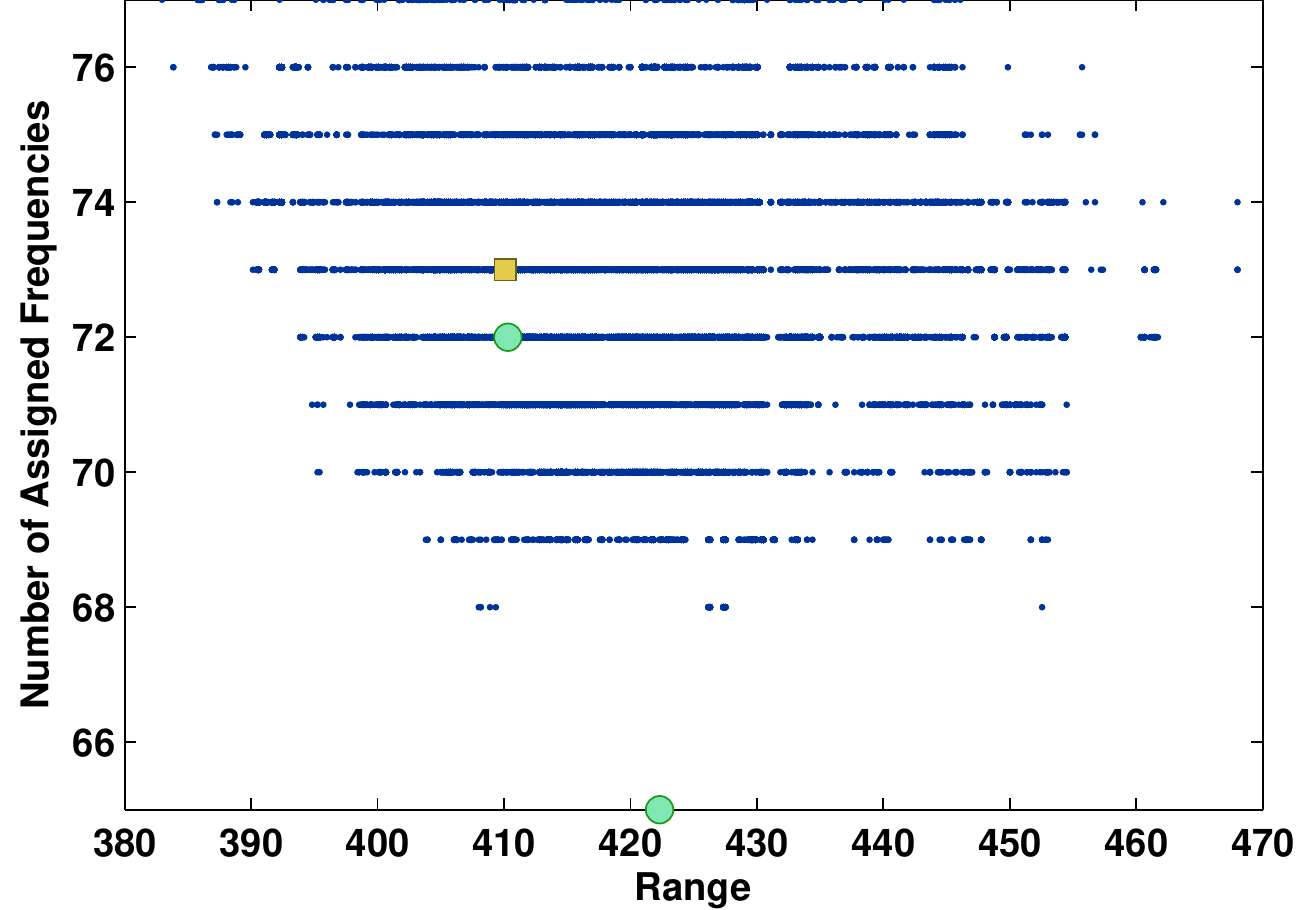}
\label{fig_hedge_enhanced_solution_space} } 
\subfigure[Enhanced and raw Hybrid algorithms' solution space in terms of range and the number of assigned frequencies. The triangle markers correspond to the solutions of the raw Hybrid algorithm with different $N_{\text{cog}}$ values (namely 140, 130, and 80). The small points correspond to the solutions acquired by randomization of the Hybrid algorithm.]
{\includegraphics[width=0.49\textwidth,keepaspectratio]{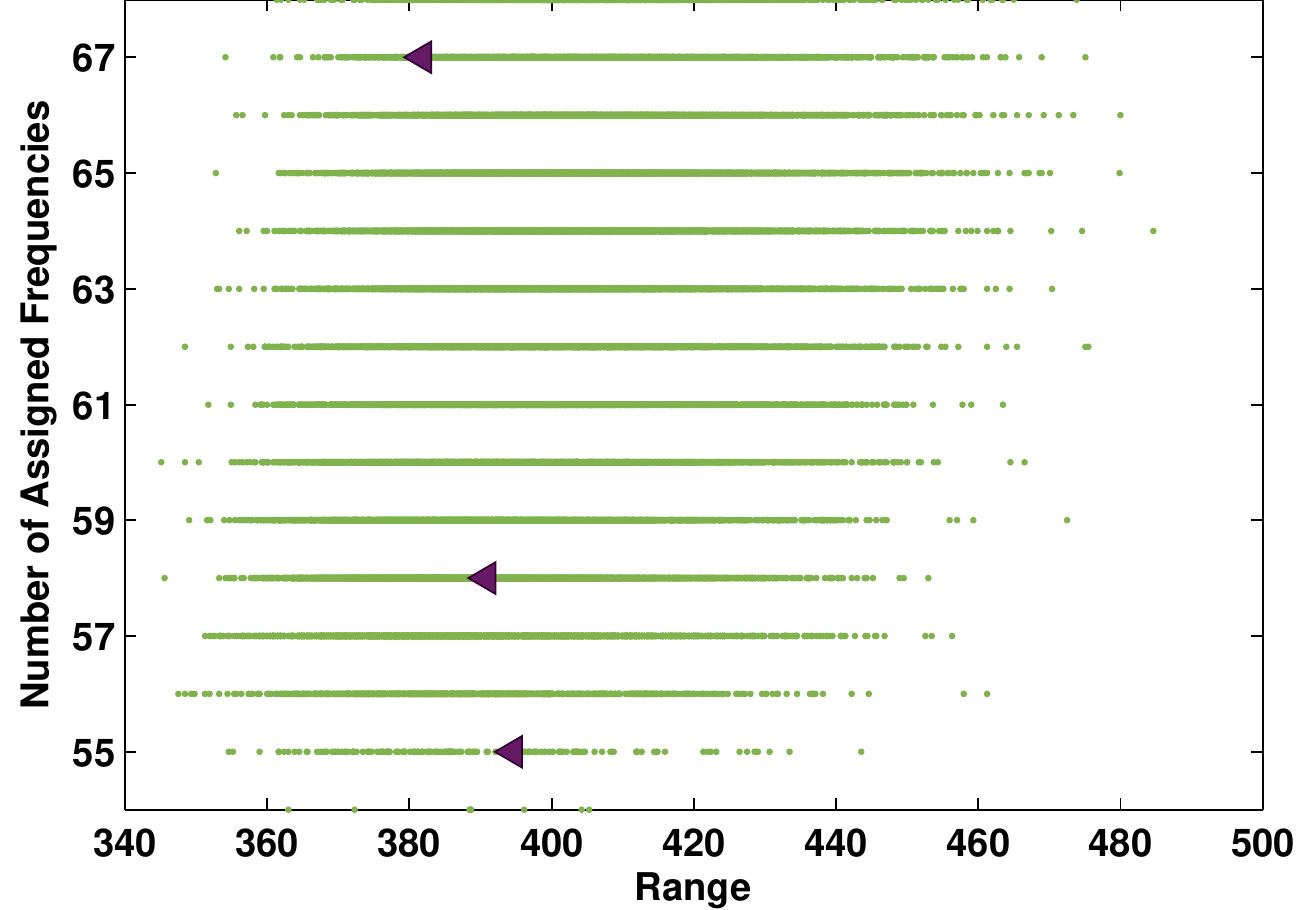}
\label{fig_hybrid_enhanced_solution_space} }
\end{center}
\caption{Solution spaces of the algorithms ($\numlink = 150$) where the lower bounds for the number of frequencies and the range are 52 and $\SI{302.2}{\mega\hertz}$.}
\label{fig_solution_spaces}
\end{figure*}
\subfigcapmargin = .0cm

\subsubsection{Analysis on Ordering Enhancement}
\label{subsec_effect_of_bf}
HEDGE is the base-line algorithm and we enhance it by randomizing the ordering. We present the genetic algorithm, the raw HEDGE algorithm, and the enhanced version results in Fig.~\ref{fig_hedge_enhanced_solution_space}. The GA improves the solution of HEDGE in terms of the number of assigned frequencies at the cost of losing the higher range of used frequencies. The GA does not improve on the range axis. The enhanced HEDGE has many different solution results and we give all possible solutions as a point on the performance metric space, namely range versus number of assigned frequencies. The left-most points correspond to the best results for a given number of assigned frequencies. With the enhanced HEDGE, we observe improvement in terms of both metrics \tbirkan{compared to raw HEDGE}.

When we order the solutions according to the number of used frequencies and range with the given order, we end up with the list of solutions in Table~\ref{tab_hedge_plus_results}. Performance metrics for GA, enhanced HEDGE, raw HEDGE algorithms\tbirkan{, and the lower bounds} are given in Table~\ref{tab_hedge_plus_results}. We want a solution with smaller $|\solnfa{i}|$ and $\rangefa(\solnfa{i})$. GA offers a better solutions in terms of the number of assigned frequencies at the cost of a higher range. On the other hand, the enhanced HEDGE has a smaller range for the same number of frequencies. Moreover, the best solution in terms of $|\solnfa{i}|$ has a smaller range compared to raw HEDGE algorithm. Hence the enhancement is significant and we can claim that the randomization enhancement (enhanced HEDGE) is more significant compared to GA enhancement. 
\begin{table}[t]
	\caption{Effect of the ordering enhancement on the HEDGE algorithm.}
	\label{tab_hedge_plus_results}
	\begin{center}
	\begin{tabular}{L{4.1cm} C{1.9cm} C{2.1cm}}
	\hline
		Solution &  $|\solnfa{i}|$ ${\mbox{(LB}=52\mbox{)}}$ & $\rangefa(\solnfa{i})$ (MHz) ${\mbox{(LB}=302.2\mbox{)}}$\\
	\hline
		Enhanced HEDGE S1  	 				&  68   & 408.0 \\
		Enhanced HEDGE S2  	 				&  69   & 403.8 \\
		Enhanced HEDGE S3  	 				&  70   & 395.3 \\
		Enhanced HEDGE S4  	 				&  71   & 394.8 \\
		Enhanced HEDGE S5  	 				&  72   & 393.9 \\
		Enhanced HEDGE S6  	 				&  73   & 390.1 \\
		Enhanced HEDGE S7  	 				&  74   & \textbf{387.3} \\
		\hline
		Genetic Algorithm S1 					&  72   & 410.3 \\
		Genetic Algorithm S2 					&  \textbf{65}   & 422.3 \\
		\hline
		HEDGE (no enhancement)  &  73   & 410.1 \\
	\hline
	\end{tabular}
	\end{center}
	\hfill\footnotesize{LB: Lower Bound}
\end{table}

Before starting to analyze the ordering enhancement on the Hybrid algorithm, we focus on the $N_{\text{cog}}$ (which determines the number of links to be assigned by the COG algorithm). Increasing $N_{\text{cog}}$ results in less frequency usage with higher $\rangefa(\solnfa{i})$. We can safely select an $N_{\text{cog}}$  between 80 and 130 for the rest of the analysis while using the order enhancement. 
\begin{table}[ht]
	\caption{Effect of $N_{\text{cog}}$ on the Hybrid algorithm.}
	\label{tab_hybrid_results}
	\centering
	\begin{tabular}{L{2.7cm} C{1.9cm} C{2.1cm}}
	\hline
		$N_{\text{cog}}$  &  $|\solnfa{i}|$ ${\mbox{(LB}=52\mbox{)}}$& $\rangefa(\solnfa{i})$ (MHz) ${\mbox{(LB}=302.2\mbox{)}}$\\
	\hline
		$140$		&  \textbf{55}   & 394.6 \\
		$130$		&  58   & 390.8 \\
		$80$		&  67   & \textbf{381.8} \\
		$50$		&  71   & 397.2 \\
		$30$		&  73   & 403.9 \\
	\hline
	\end{tabular}
\end{table}

With appropriate choices of $N_{\text{cog}}$, Hybrid algorithm solutions are better than raw HEDGE and enhanced HEDGE algorithms in terms of range and the number of used frequencies. Even without utilizing the randomization enhancement, the Hybrid algorithm has promising solutions, as can be seen in Table~\ref{tab_hybrid_results}.

When we consider randomization enhancement, we end up with different solution results. In Fig.~\ref{fig_hybrid_enhanced_solution_space}, each point corresponds to a different solution, represented by its performance metrics --the range and number of frequencies used. Three solutions obtained from the raw Hybrid algorithm with $N_{\text{cog}}$ values 80, 130, and 140 are also presented in Fig.~\ref{fig_hybrid_enhanced_solution_space} with a triangle marker. The enhanced version has many solution results and for each value in the y-axis, the left-most point corresponds to the best result for that specific row. With randomized ordering enhancement, we observe a significant improvement on the Hybrid algorithm.  Table~\ref{tab_hybrid_plus_results} presents the $|\solnfa{i}|$ and $\rangefa(\solnfa{i})$ metrics for representative solutions of the enhanced Hybrid and the raw Hybrid algorithms. \tbirkan{Note that the performance metrics of the enhanced Hybrid algorithm are close to lower bound values.}
\begin{figure}[t]
	\centering
	\includegraphics[width=0.99\columnwidth,keepaspectratio]
	{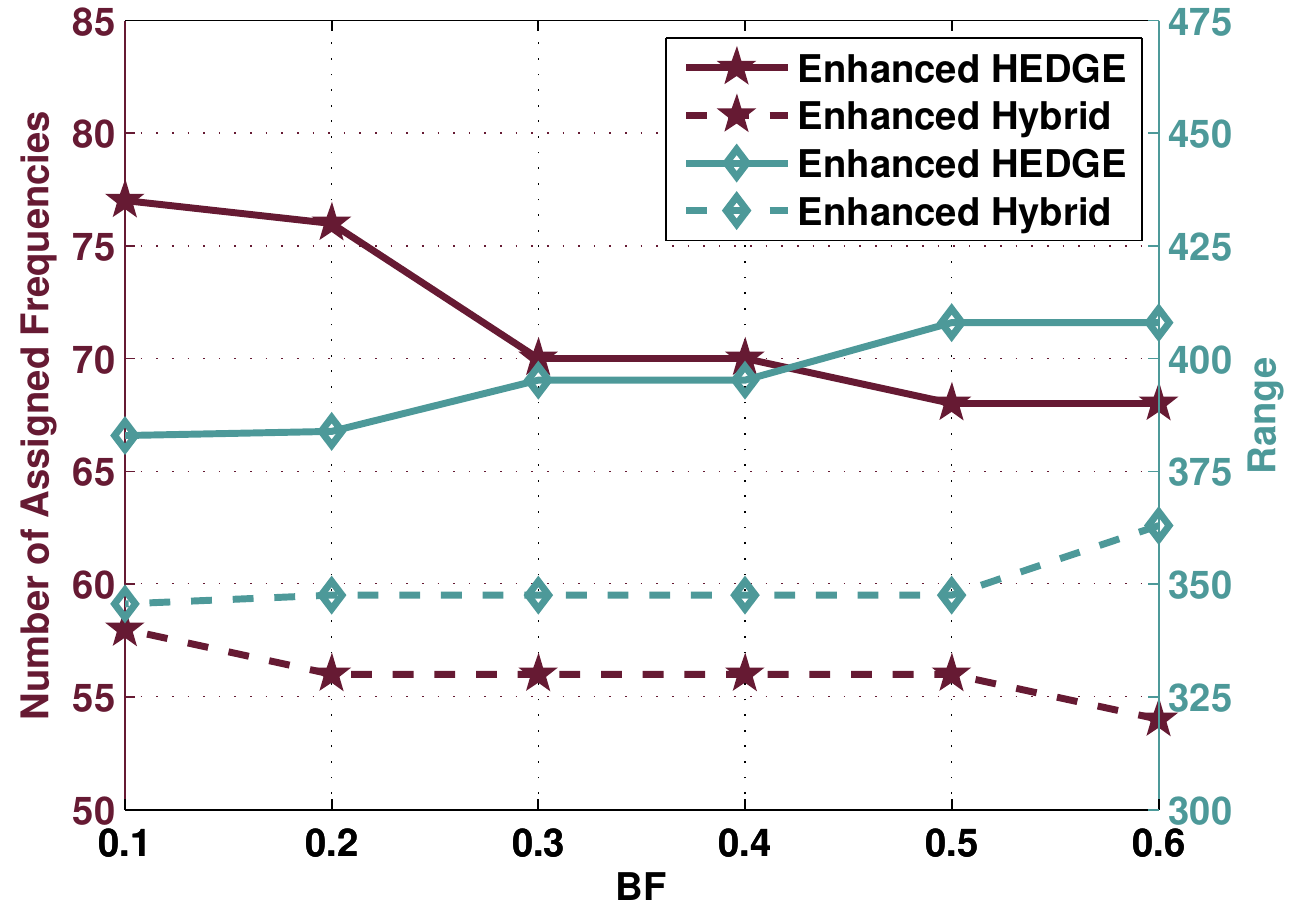}
	\caption{Effect of balancing factor on the enhanced HEDGE and Hybrid algorithms where the lower bounds for the number of frequencies and the range are 52 and $\SI{302.2}{\mega\hertz}$. The left axis and the curves with star markers show the effect on the number of assigned frequencies. The right axis and the curves with diamond marker show the effect on the range. Solid curves correspond to the enhanced HEDGE algorithm while dashed curves correspond to the enhanced Hybrid algorithm.}
	\label{fig_effectBF}
\end{figure}
\begin{table}[ht]
	\caption{Effect of the ordering enhancement on the Hybrid algorithm.}
	\label{tab_hybrid_plus_results}
	\centering
	\begin{tabular}{L{3.4cm} C{1.9cm} C{2.1cm}}
	\hline
		Solution &  $|\solnfa{i}|$ ${\mbox{(LB}=52\mbox{)}}$ & $\rangefa(\solnfa{i})$ (MHz) ${\mbox{(LB}=302.2\mbox{)}}$\\
	\hline
		Enhanced Hybrid S1  	 			&  \textbf{54}   & 363.0 \\
		Enhanced Hybrid S2  	 			&  55   & 354.6 \\
		Enhanced Hybrid S3  	 			&  56   & 347.6 \\
		Enhanced Hybrid S4  	 			&  58   & \textbf{345.6} \\
	\hline
		Hybrid ($N_{\text{cog}}=140$)  &  55   & 394.6 \\
		Hybrid ($N_{\text{cog}}=130$)  &  58   & 390.8 \\
		Hybrid ($N_{\text{cog}}=80$)   &  67   & 381.8 \\	
	\hline
	\end{tabular}
\end{table}

Figure~\ref{fig_effectBF} shows the effect of the balancing factor on the enhanced HEDGE and Hybrid algorithms. The enhanced HEDGE and Hybrid algorithms have many solution results and ranking them according to \eqref{eq_scoring_function} gives us the best solution with the given parameters. Increasing the balancing factor parameter means increasing the weight of the number of assigned frequencies. When we consider the solid and dashed curves as two different groups, we can see that the dashed curves (i.e., curves of the enhanced Hybrid algorithm) offer a better solution with a  smaller range and fewer number of assigned frequencies. When we focus on the curves with the same markers, we observe the affect of the balancing factor on the corresponding metric. Increasing the balancing factor results in a lesser number of assigned frequencies at the cost of the increased range.

\subsubsection{Analysis on Effect of $\numlink$}
\label{subsec_effect_of_nl}
In Figs.~\ref{fig_effectNL_onAF} and~\ref{fig_effectNL_onRange}, we analyze the effect of the number of links on system performance metrics. The number of assigned frequencies is presented in Fig.~\ref{fig_effectNL_onAF} and the range of the assignment is presented in Fig.~\ref{fig_effectNL_onRange}. In both of the figures, we selected the balancing factor as 0.4.

In Fig.~\ref{fig_effectNL_onAF}, the effect of $\numlink$ on the number of assigned frequencies is shown for the HEDGE, enhanced HEDGE, and enhanced Hybrid algorithms. The HEDGE algorithm, as a baseline greedy algorithm, has the highest number of assigned frequencies. Hence, it requires more resources than other algorithms. The enhanced HEDGE algorithm is slightly better than the HEDGE algorithm, while the enhanced Hybrid algorithm has a significant advantage over both of them. \tbirkan{Moreover, the enhanced Hybrid algorithm finds, in a time-efficient manner, solutions close to the lower bound values.} With the given parameters, the enhanced Hybrid algorithm requires 15 fewer frequency bands on average than the HEDGE algorithm. Another observation is that the gain is increasing as we increase the $\numlink$ (e.g., for 100 links, the gain, in terms of number of required frequency bands, is nine while the improvement for 200 links is 22).
\begin{figure}[t]
	\centering
	\includegraphics[width=0.99\columnwidth,keepaspectratio]
	{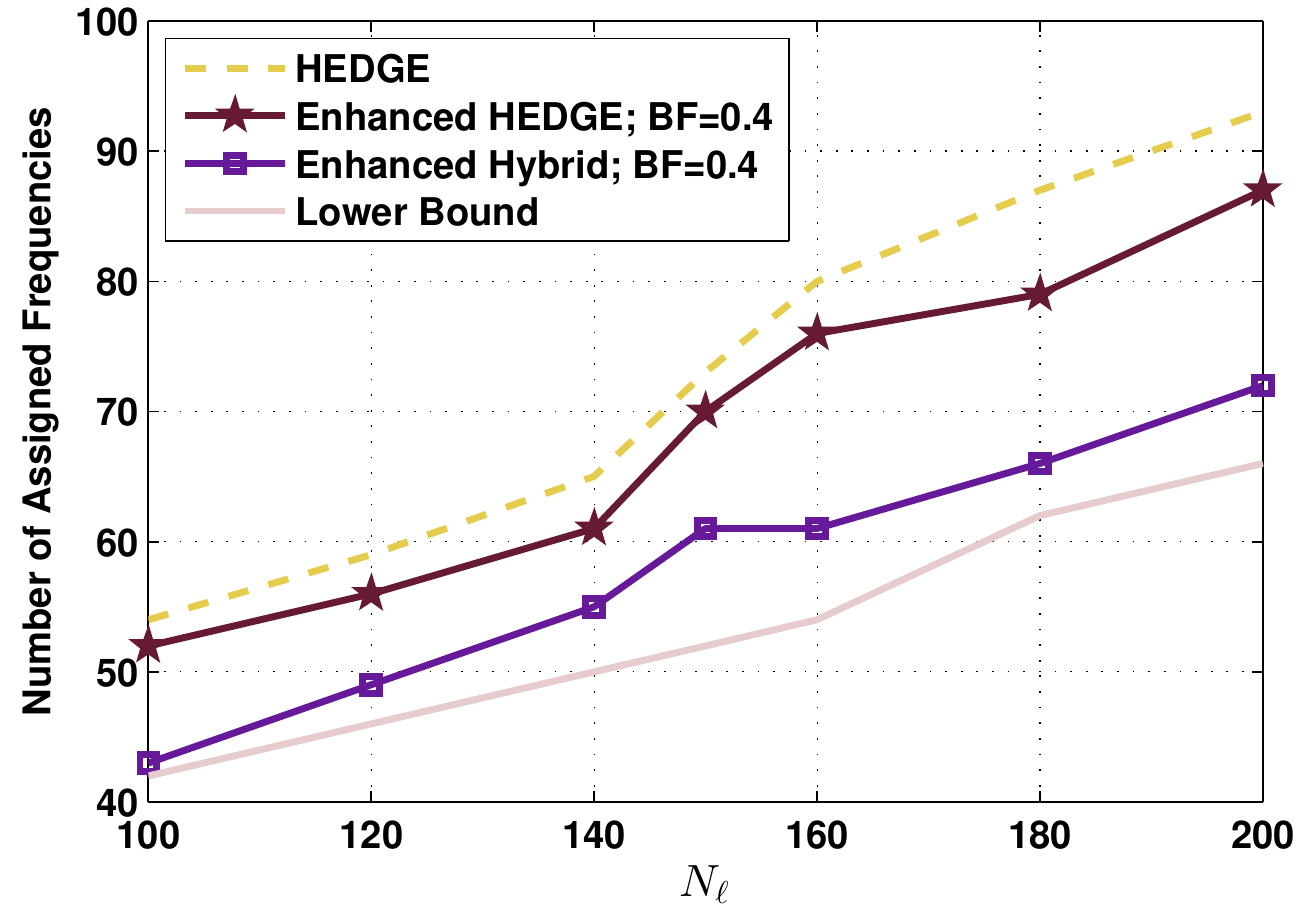}
	\caption{Effect of $\numlink$ on the number of assigned frequencies.}
	\label{fig_effectNL_onAF}
\end{figure}

In Fig.~\ref{fig_effectNL_onRange}, the effect of $\numlink$ on the range of assigned frequencies is presented for the HEDGE, enhanced HEDGE, and enhanced Hybrid algorithms. Having a smaller range is better for re-utilizing and managing the resources at close locations. The HEDGE algorithm requires larger range values than do the enhanced HEDGE and Hybrid algorithms. The enhanced HEDGE algorithm is slightly better than the HEDGE algorithm, while a significant enhancement is achieved by the enhanced Hybrid algorithm. Also we see that the improvement multiplies when we increase $\numlink$. \tbirkan{Moreover, the enhanced Hybrid algorithm is close enough to the lower bound values with a reasonable run time that it will be considered after analyzing the measured data}.
\begin{figure}[t]
	\centering
	\includegraphics[width=0.99\columnwidth,keepaspectratio]
	{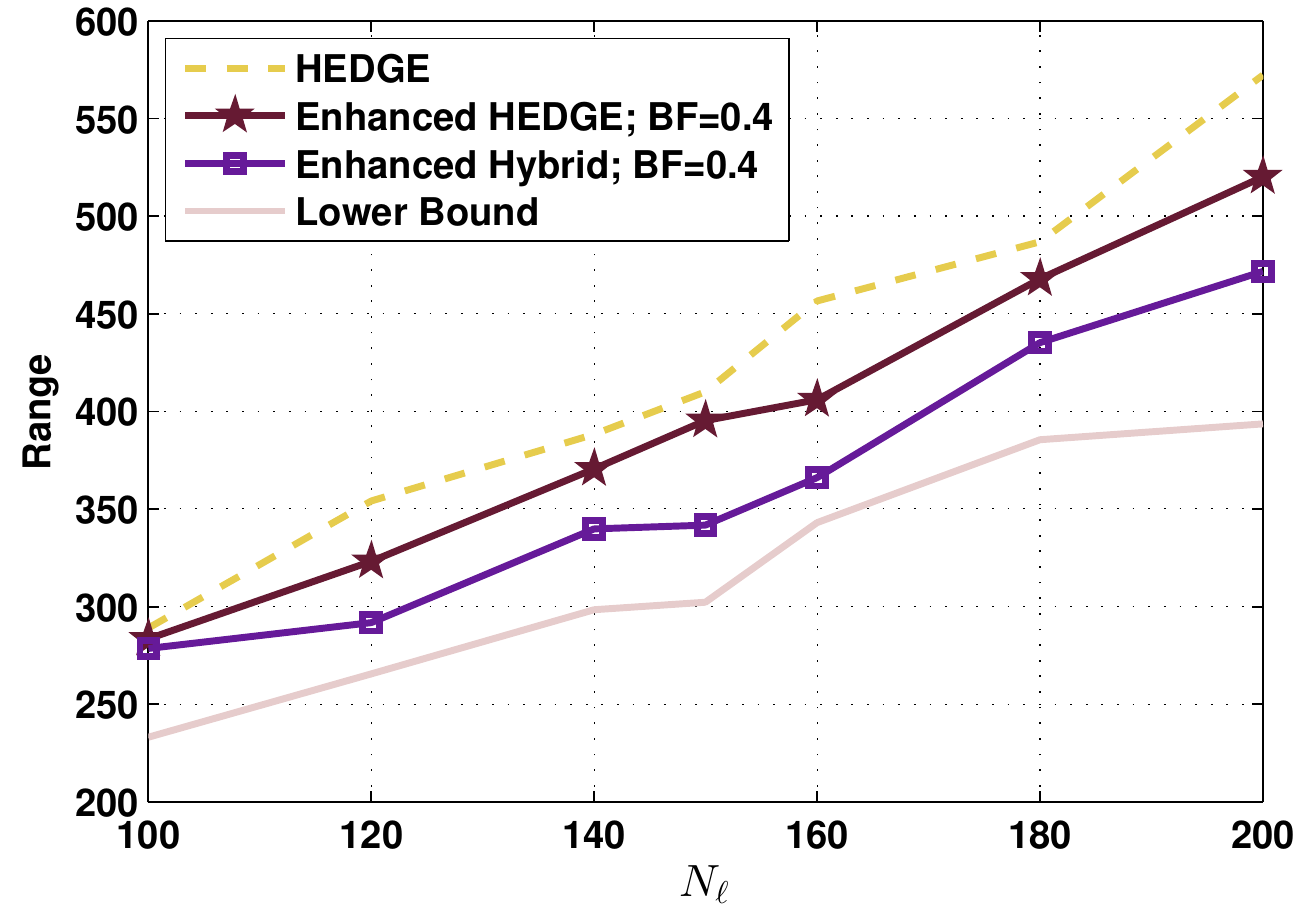}
	\caption{Effect of $\numlink$ on the range of assigned frequencies.}
	\label{fig_effectNL_onRange}
\end{figure}

\subsection{Analysis with Measurement Data}

The topology of the nodes is shown in Fig.~\ref{fig_simulator}, which is a region in South Korea's Daejeon. The links are placed in a more separated manner to reduce interference constraints. The frequency assignment solutions for the measured data for $\numlink=200$ and $\numlink=50$ are given in Table~\ref{tab_results_of_measured_data}. 

The HEDGE algorithm performs poorly compared to the proposed algorithms. \tbirkan{For $\numlink=200$ and $\numlink=50$, the lower bounds for the required number of frequencies are 5 and 4, respectively.} For both $\numlink=200$ and $\numlink=50$ cases, the enhanced HEDGE algorithm requires fewer frequencies and a smaller range \tbirkan{compared to raw HEDGE} (gain is at least  20\%). The enhanced Hybrid algorithm achieves more than 50\% enhancement in terms of the number of assigned frequencies at the cost of 25\% range increments for $\numlink=200$. In the $\numlink=50$ case, the enhanced Hybrid algorithm achieves nearly 40\% and 30\% improvement in terms of the number of assigned frequencies and the range, respectively. Hence, these results validate our proposed algorithms also with the measured data and the realistic antenna features. \tbirkan{Also note that, for the measured data the proposed algorithms achieve very close to the lower bound and specifically the same for $the \numlink=50$ case.} 

\begin{table}[t]
	\caption{Effect of the ordering enhancement on the HEDGE and the Hybrid algorithm with the measured data.}
	\label{tab_results_of_measured_data}
	\centering
	\begin{tabular}{C{0.4cm} L{3.3cm} C{1.1cm} C{2.1cm}}
	\hline
		& Solution &  $|\solnfa{i}|$ & $\rangefa(\solnfa{i})$ (MHz) \\
	\hline
		\parbox[t]{0.1cm}{\multirow{3}{*}{\rotatebox[origin=c]{90}{$\numlink=200\;\;$ }}}& HEDGE   &  14   & 201.6 \\
		 & Enhanced HEDGE S1 			  &  10  & \textbf{160.8} \\
		 & Enhanced HEDGE S2 			  &  9  & 162.0 \\
		 & Enhanced Hybrid S1 			  &  \textbf{6}  & 259.1 \\
	\hline
		\parbox[t]{0pt}{\multirow{3}{*}{\rotatebox[origin=c]{90}{$\numlink=50\;\;$ }}}& HEDGE   &  7   & 196.2 \\
		& Enhanced HEDGE S1 			&  6  & 141.6 \\
		& Enhanced HEDGE S2 			&  5  & 142.5 \\ 
		& Enhanced Hybrid S1 		&  \textbf{4}   & \textbf{133.9} \\
	\hline
	\end{tabular}
\end{table}

\subsection{\tbirkan{Time Analysis}}
\tbirkan{We also analyzed the runtime costs of the proposed algorithms and a straight forward tabu search algorithm. For the analysis we used modified CELAR data and specifically the scen02 data is used. There are 200 links and 1235 constraints for the links. The minimum number of required frequencies that is found in the literature is 14 while the lower bound is 13~\cite{tiourine1995overviewOA,kolen1994constraintSA,bouju1995intelligentSF}.}

\tbirkan{Figure~\ref{fig_time_cost} displays the time-limited run results. Under a limited time assumption, due to the random nature of the algorithms, we replicated the runs and at each run we ended up with different solutions-- sometimes with bad ones. Curves correspond to the mean value of the number of used frequencies with showing the maximum and the minimum values as bars around the mean. When we increase the time limit is increased, the achieved solutions tend to improve. However, prior to 500 seconds, a Tabu search generally leads to no feasible solutions. Since our aim here is to develop an effective algorithm for easy integration into existing software and achieving reasonable results in a short duration, we have focused on sequential greedy algorithms. This analysis justifies our arguments. }
\begin{figure}[t]
	\centering
	\includegraphics[width=0.99\columnwidth,keepaspectratio]
	{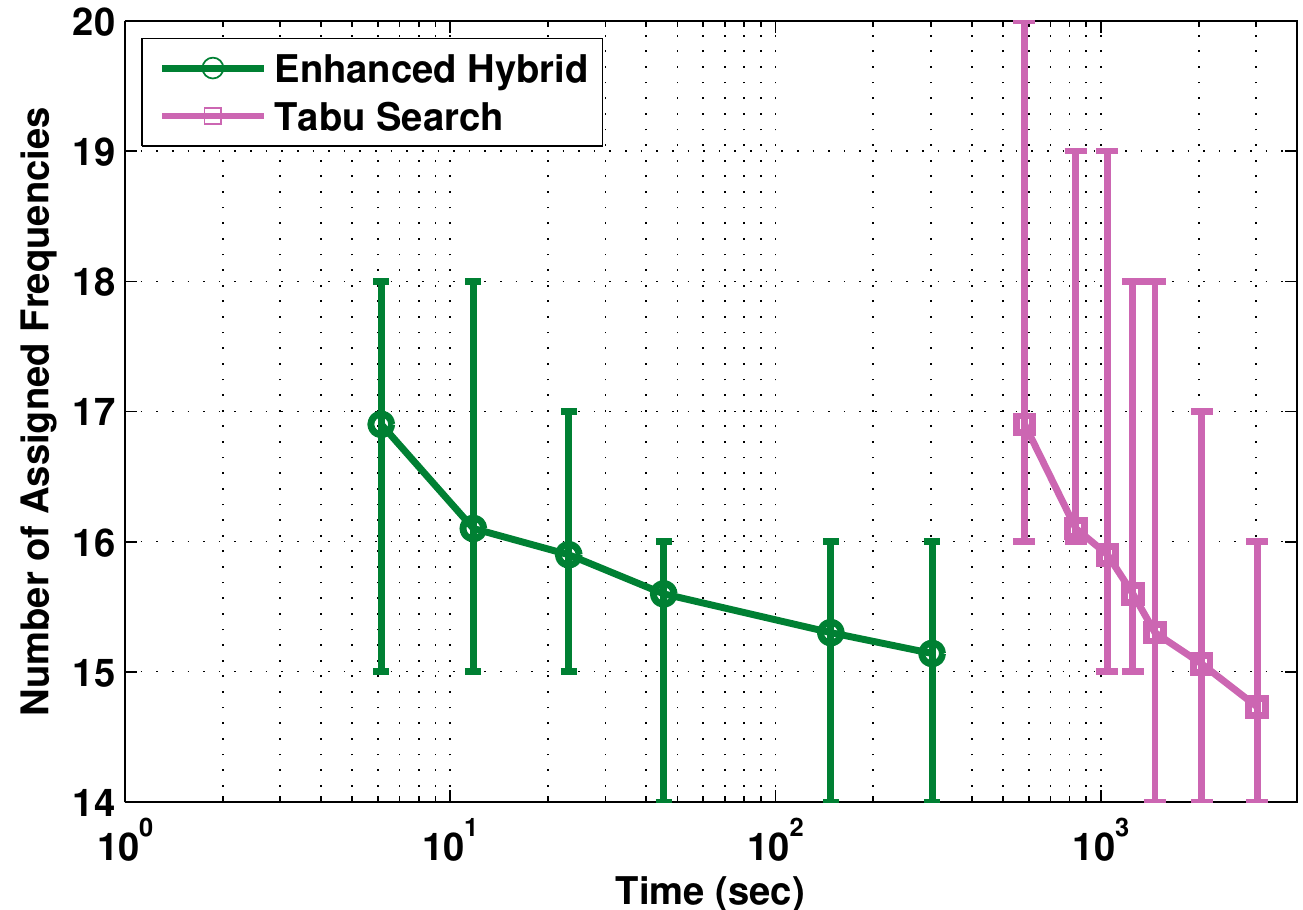}
	\caption{Performance in terms of number of assigned frequencies for a time limited scenario. Each method is replicated 20 times with different time limitations and the performance metric is averaged over run results.}
	\label{fig_time_cost}
\end{figure}

\vspace{10pt}
\section{\uppercase{Conclusion}}
\label{sec:concl}
In this paper, we proposed frequency resource allocation algorithms by using \tbirkan{simple} randomization \tbirkan{step} to enhance greedy heuristics. The proposed algorithms were compared with greedy heuristics and genetic algorithm in terms of the number of assigned frequencies and the range of the assignment. Our analysis was carried out using synthetic data and measured data. First, we observed that the conventional genetic algorithm's enhancement is not significant in terms of range. The enhanced HEDGE algorithm significantly improved the solution by reducing the number of assigned frequencies by 7\% for nearly the same range with $\numlink=150$. If we compare an equal number of assigned frequencies, we confirm a significant improvement in terms of the range. Similarly, the enhanced Hybrid algorithm significantly reduced the number of assigned frequencies and the improvement in terms of the range of assignment was 11-13\% for the same number of assigned frequencies with $\numlink=150$. Moreover, if we consider the enhanced Hybrid and the base-line algorithm (HEDGE), then increasing the number of links resulted in nearly 20\% improvement in terms of the number of assigned frequencies. We also observed similar enhancements when we used the measured data that includes the effects of the 3-D terrain. \tbirkan{Moreover, we compared the time cost of our proposed methods with a Tabu search, and found our motivation to design simple and effective randomized greedy algorithms to be justified.} Therefore, by incorporating the randomization, the proposed algorithms significantly enhanced the heuristics for the FAP. Future work will consider more general system configurations such as multiple antennas and polarization.


\bibliographystyle{ieeetr}
\bibliography{bib_jcn_fap}

\epsfysize=3.2cm
\begin{biography}{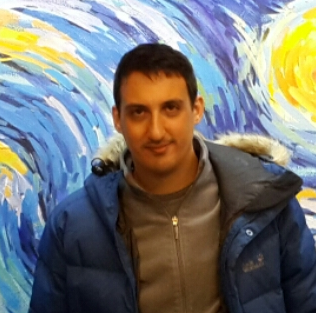}{H. Birkan Yilmaz}
received his M. Sc. and Ph.D. degrees in Computer Engineering from Bogazici University in 2006 and 2012, respectively, his B.S. degree in Mathematics in 2002. Currently, he works as a post-doctoral researcher at Yonsei Institute of Convergence Technology, Yonsei University, Korea. He holds TUBITAK National Ph.D. Scholarship and is a member of TMD (Turkish Mathematical Society). His research interests include cognitive radio, spectrum sensing, molecular communications, and detection and estimation theory.
\end{biography}

\epsfysize=3.2cm
\begin{biography}{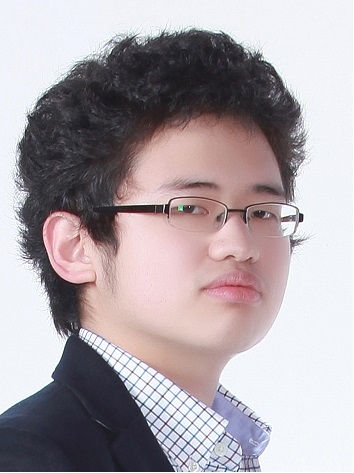}{Bon-Hong Koo}
received his B.S. degree in the School of Integrated Technology from Yonsei University, Korea, in 2014. He is now with the School of Integrated Technology, at the same university and is working toward the Ph.D. degree. His research interest includes emerging communication technologies. \\
\end{biography}

\epsfysize=3.0cm
\begin{biography}{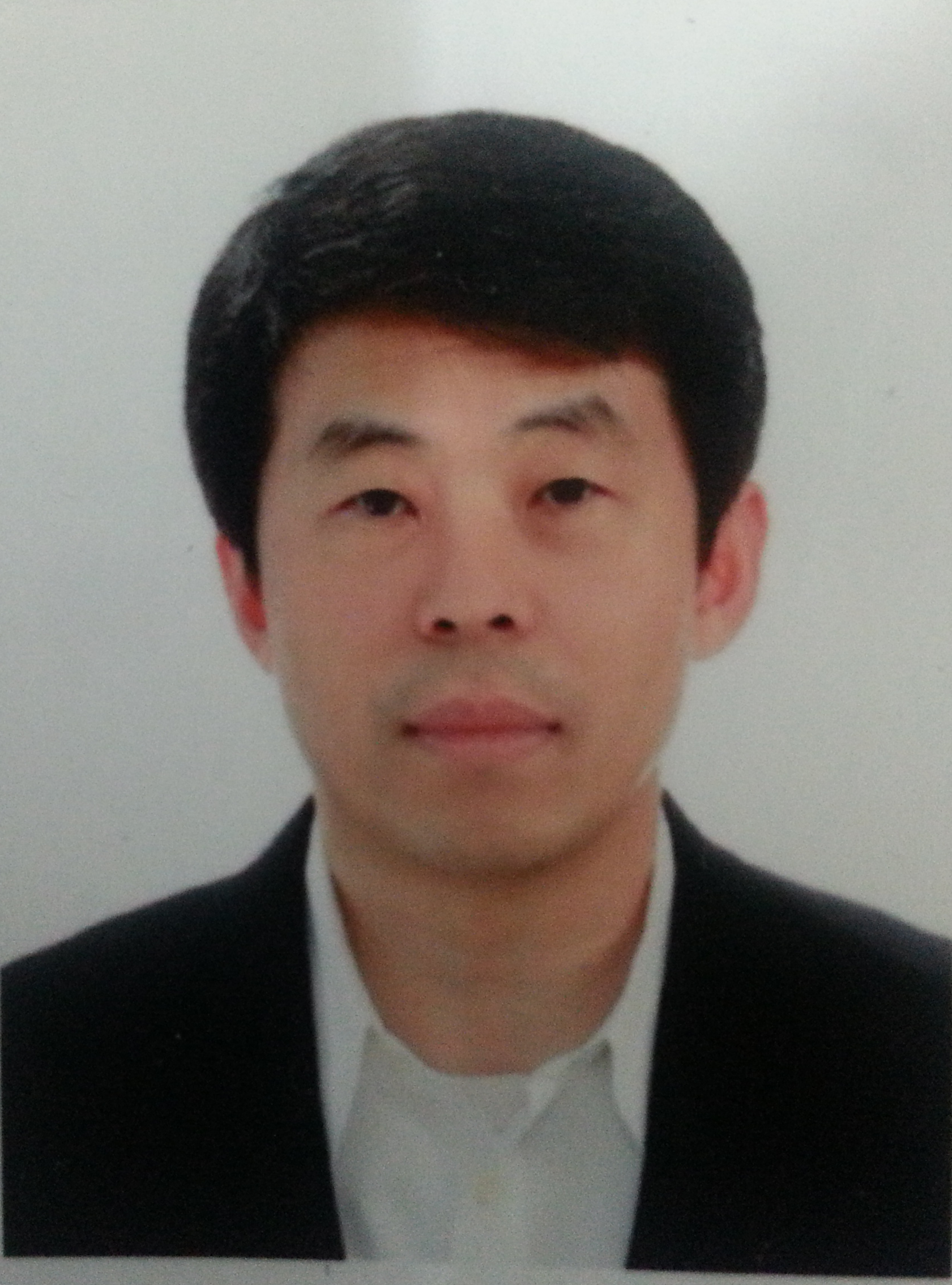}{Sung-Ho Park} was born in Seoul, Korea in 1962. He received his B.S. degree in Mathematics from Korea University, Seoul, Korea in 1984. He is now with Open SNS in Seoul, Korea as Vice President. His research interest includes radio propagation for military communications. \\ \\
\end{biography}

\epsfysize=3.0cm
\begin{biography}{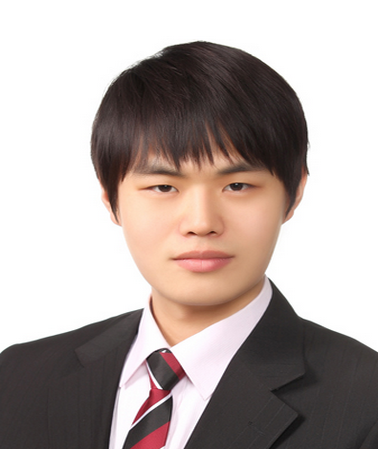}{Hwi-Sung Park}
received the B.S. degrees in Electronic and Electrical Engineering from Sungkyun-kwan University, Korea in 2012, and the M.S. degree in Electrical and Electronic Engineering from KAIST, Korea in 2014. He is currently a researcher in the Department of the 2nd R\&D Institute-1, Agency for Defense Development. His research has been focused on spectrum management in tactical networks and link adaptation in multi-RAT networks.
\end{biography}

\epsfysize=3.0cm
\begin{biography}{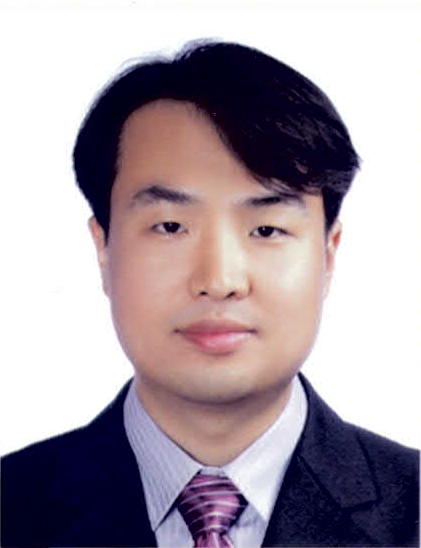}{Jae-Hyun Ham}
received his B.S. degree in Computer Science and Engineering from Dongguk University, Korea, in 1999, and his M.S. degree in Computer Science and Engineering from POSTECH, Korea, in 2001. He joined the Agency for Defense Development, Korea, in 2001, where he is working currently as a senior researcher in the Department of the 2nd R\&D Institute-1. His research interests include tactical network \& spectrum management, mobile ad hoc network, and traffic monitoring and analysis.
\end{biography}

\epsfysize=2.9cm
\begin{biography}{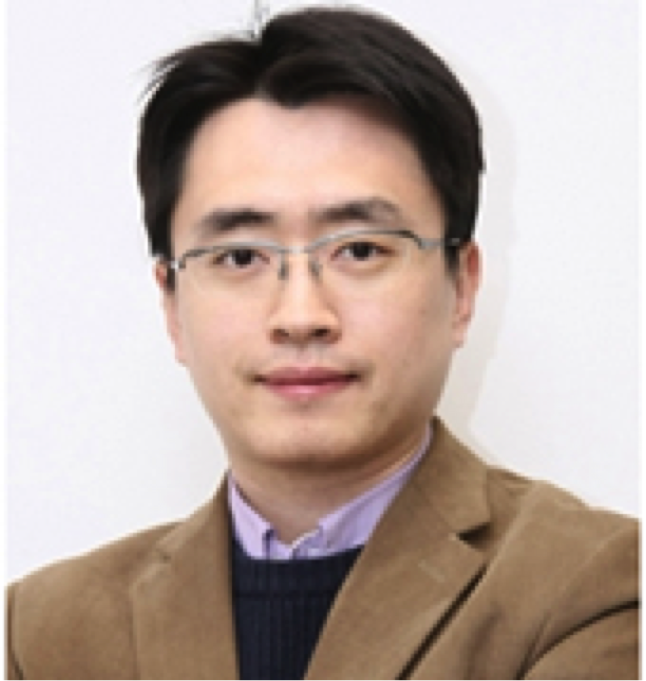}{Chan-Byoung Chae}
('S06-'M09-'SM12) is Associate Professor in the School of Integrated Technology, College of Engineering, Yonsei University, Korea. He was a Member of Technical Staff (Research Scientist) at Bell Laboratories, Alcatel-Lucent, Murray Hill, NJ, USA from June 2009 to Feb 2011. Before joining Bell Laboratories, he was with the School of Engineering and Applied Sciences at Harvard University, Cambridge, MA, USA as a Post-Doctoral Research Fellow. He received the Ph.D. degree in Electrical and Computer Engineering from The University of Texas (UT), Austin, TX, USA in 2008, where he was a member of the Wireless Networking and Communications Group (WNCG). Prior to joining UT, he was a Research Engineer at the Advanced Research Lab., the Telecommunications R\&D Center, Samsung Electronics, Suwon, Korea, from 2001 to 2005. He was a Visiting Scholar at the WING Lab, Aalborg University, Denmark in 2004 and at University of Minnesota, MN, USA in August 2007. While having worked at Samsung, he participated in the IEEE 802.16e standardization, where he made several contributions and filed a number of related patents from 2004 to 2005. His current research interests include capacity analysis and interference management in energy-efficient wireless mobile networks and nano (molecular) communications. 

He has served/serves as an Editor for the IEEE Trans. on Wireless Communications, IEEE Trans. Molecular, Biological, and Multi-scale Comm., IEEE Trans. on Smart Grid, IEEE ComSoc Technology News, and IEEE/KICS Jour. of Comm. Networks. He was a Guest Editor for the IEEE Jour. Sel. Areas in Comm. (special issue on molecular, biological, and multi-scale comm.). He is an IEEE Senior Member.
Dr. Chae was the recipient/co-recipient of the IEEE INFOCOM Best Demo Award (2015), the IEIE/IEEE Joint Award for Young IT Engineer of the Year (2014), the KICS Haedong Young Scholar Award (2013), the IEEE Signal Processing Magazine Best Paper Award (2013), the IEEE ComSoc AP Outstanding Young Researcher Award (2012), the IEEE Dan. E. Noble Fellowship Award (2008), two Gold Prizes (1st) in the 14th/19th Humantech Paper Contest, and the KSEA-KUSCO scholarship (2007). He also received the Korea Government Fellowship (KOSEF) during his Ph.D. studies.
\end{biography}
\end{document}